\documentclass{llncs}

\usepackage[utf8]{inputenc}
\usepackage{microtype}
\usepackage{xspace}
\usepackage{balance}
\usepackage{hyperref}
\usepackage{url}
\usepackage{breakurl}
\usepackage{enumerate}
\usepackage{enumitem}

\usepackage{listings}
\lstdefinelanguage{program}{%
  keywords={%
    let,pass,function,%
    var,const,bool,int,void,atomic,%
    while,do,if,then,else,assume,assert,call,return,rule,forall,with,new,choose,skip,%
    task,async,yield,for,wait%
  },
  morecomment=[l]{//},
  morecomment=[s]{/*}{*/},
  morecomment=[n]{(*}{*)},
  mathescape=true,
  escapeinside=`',
}
\usepackage{multirow}

\usepackage{graphicx}
\usepackage{tikz}
\usetikzlibrary{positioning,shapes,shadows,arrows}
\usetikzlibrary{decorations.pathmorphing,snakes}
\usepackage{tkz-graph}
\usepackage{moresize}

\usepackage{at}

\usepackage{amsmath}
\usepackage{amsfonts}
\usepackage{amssymb}
\usepackage{latexsym}
\usepackage{wasysym} 
\usepackage{stmaryrd}
\usepackage{alltt}
\usepackage{mathpartir}
\usepackage[ligature,reserved,inference]{semantic}
\usepackage[lined]{algorithm2e}
\usepackage{graphicx}
\usepackage{pgf}
\usepackage{tikz}
\usepackage{multicol}
\usepackage{fancyvrb}
\usepackage{rotating}
\usepackage{latexsym}
\usepackage{tabulary}
\usepackage{wrapfig}
\usepackage{caption,subcaption}

\usetikzlibrary{arrows,automata}
\usetikzlibrary{positioning,shapes,shadows}
\usetikzlibrary{decorations}
\makeatletter

\newcommand{\angles}[1]{\ensuremath{\left\langle {#1} \right\rangle}}

\let\oldae=\ae
\renewcommand{\ae}{\oldae\xspace}

\newcommand{\set}[1]{{\{ #1 \}}}

\newcommand{\tup}[1]{\angles{#1}}

\mathlig{,.,}{,\ldots,}
\mathlig{,..}{,\ldots}
\mathlig{..,}{\ldots,}
\mathlig{-->}{\longrightarrow}
\mathlig{<=>}{\Leftrightarrow}
\mathlig{==>}{\implies}
\mathlig{|->}{\mapsto}
\mathlig{<-}{\gets}
\mathlig{->}{\mathbin{\rightarrow}}
\mathlig{~>}{\mathbin{\leadsto}}
\mathlig{..}{\ldots}
\mathlig{/|}{\land}
\mathlig{|/}{\lor}
\mathlig{=>}{\mathbin{\Rightarrow}}
\mathlig{<=}{\leq}
\mathlig{==}{\equiv}
\mathlig{~=}{\approx}
\mathlig{!=}{\neq}
\mathlig{|-}{\vdash}
\mathlig{|=}{\models}
\mathlig{~}{\sim}

\newcommand{\asgreekstyle}{\mathrm} 
\newcommand{\greek}[1]{\ensuremath{#1}\xspace}
\newatcommand a {\greek{\alpha}}
\newatcommand b {\greek{\beta}}
\newatcommand c {\greek{\xi}}
\newatcommand C {\greek{\Xi}}
\newatcommand d {\greek{\delta}}
\newatcommand D {\greek{\Delta}}
\newatcommand e {\greek{\varepsilon}}
\newatcommand E {\greek{\asgreekstyle{E}}}
\newatcommand f {\greek{\phi}}
\newatcommand F {\greek{\Phi}}
\newatcommand g {\greek{\gamma}}
\newatcommand G {\greek{\Gamma}}
\newatcommand h {\greek{\varrho}}
\newatcommand j {\greek{\iota}}
\newatcommand k {\greek{\kappa}}
\newatcommand l {\greek{\ell}\xspace}
\newatcommand ll {\greek{\lambda}}
\newatcommand L {\greek{\Lambda}}
\newatcommand m {\greek{\mu}}
\newatcommand n {\greek{\eta}}
\newatcommand o {\greek{\theta}}
\newatcommand oo {\greek{\vartheta}}
\newatcommand O {\greek{\Theta}}
\newatcommand p {\greek{\varphi}}
\newatcommand P {\greek{\Pi}}
\newatcommand q {\greek{\epsilon}}
\newatcommand Q {\greek{\Phi}} 
\newatcommand r {\greek{\rho}}
\newatcommand s {\greek{\sigma}}
\newatcommand ss {\greek{\sigma}}
\newatcommand S {\greek{\Sigma}}
\newatcommand t {\greek{\tau}}
\newatcommand u {\greek{\upsilon}}
\newatcommand U {\greek{\Upsilon}}
\newatcommand v {\greek{\nu}}
\newatcommand w {\greek{\omega}}
\newatcommand W {\greek{\Omega}}
\newatcommand x {\greek{\chi}}
\newatcommand y {\greek{\psi}}
\newatcommand Y {\greek{\Psi}}
\newatcommand z {\greek{\zeta}}

\newatcommand 0 {\emptyset}
\newatcommand \ {\setminus}

\newcommand{\mathfnstyle}[1]{\ensuremath{\mathrm{#1}}}
\reservestyle{\mathfn}{\mathfnstyle}

\newcommand{\metalangkeywordstyle}[1]{\ensuremath{\mathsf{#1}}}
\reservestyle{\metalangkeyword}{\metalangkeywordstyle}

\newcommand{\semanticdomainstyle}[1]{%
  \ensuremath{\mathchoice%
    {\mbox{\normalfont\ensuremath{#1}}}%
    {\mbox{\normalfont\ensuremath{#1}}}%
    {\mbox{\normalfont\scriptsize\ensuremath{#1}}}%
    {\mbox{\normalfont\tiny\ensuremath{#1}}}}}

\reservestyle{\semanticdomain}{\semanticdomainstyle} 

\mathfn{min,max,log,dom,range}

\semanticdomain{%
  Bools[\mathbb{B}],
  Nats[\mathbb{N}],
  Lab[\mathbb{A}],
  Ints[\mathbb{Z}]
}

\semanticdomain{%
  Vals[\mathbb{V}],
  Labels[\mathbb{L}],
  Writes[\mathbb{W}],
  Reads[\mathbb{R}],
  Events[\mathbb{E}v],
  Execs[\mathbb{E}xec],
  Traces[\mathbb{T}],
  Ids[\mathbb{P}],
  Pids[\mathbb{P}],
  Mids[\mathbb{M}],
  Plds[\mathbb{V}],
  Val[\mathbb{V}],
  Lsts[L],
  Msgs[M],
  AdtM[\mathbb{T}],
}

\metalangkeyword{true,false}
\metalangkeyword{ret,op,msg,src,dst,pay,proc,dest}
\mathfn{past,future,width}
\metalangkeyword{push,pop}
\metalangkeyword{read,write,lock,unlock}

\newcommand{\match}{\mathbin{\mapstochar\relbar\mapsfromchar}}

\newcommand{\send}[2]{\mathrm{send}_{#1}({#2})}
\newcommand{\rec}[2]{\mathrm{rec}_{#1}({#2})}

\newcommand{\senda}[1]{\mathrm{send}({#1})}
\newcommand{\reca}[1]{\mathrm{rec}({#1})}
\newcommand{\paral}{\,||\,}

\newcommand\asynchTr[1]{\mathrm{asTr}({#1})} 
\newcommand\asynchExec[1]{\mathrm{asEx}({#1})} 
\newcommand\synchTr[2]{\mathrm{sTr}_{#2}({#1})} 
\newcommand\synchExec[2]{\mathrm{sEx}_{#2}({#1})} 
\newcommand{\xRightarrow}[2][]{\ext@arrow 0359\Rightarrowfill@{#1}{#2}}

\newcommand\asynchSt[1]{\mathrm{asSt}({#1})} 

\title{On the Completeness of Verifying Message Passing Programs under Bounded Asynchrony\thanks{This work is supported in part by the European Research Council (ERC) under the European Union's Horizon 2020 research and innovation programme (grant agreement No 678177).}}
\author{Ahmed Bouajjani\inst{1} \and Constantin Enea\inst{1} \and Kailiang Ji\inst{1} \and Shaz Qadeer\inst{3} }
\institute{IRIF, University Paris Diderot \& CNRS, 
\email{\{abou,cenea,jkl\}@irif.fr},
\and Microsoft Research, 
\email{qadeer@microsoft.com}}

\begin{document}

\pagestyle{headings}
\bibliographystyle{splncs03}




\maketitle

\vspace{-5mm}
\begin{abstract}
We address the problem of verifying message passing programs, defined as a set of parallel processes communicating through unbounded FIFO buffers.
We introduce a bounded analysis that explores a special type of computations, called $k$-synchronous. These computations can be viewed as (unbounded) sequences of interaction phases, each phase allowing at most $k$ send actions (by different processes), followed by a sequence of receives corresponding to sends in the same phase. We give a procedure for deciding {\em $k$-synchronizability} of a program, i.e., whether every computation is equivalent (has the same happens-before relation) to one of its $k$-synchronous computations. We also show that 
reachability over $k$-synchronous computations and checking $k$-synchronizability are both PSPACE-complete. Furthermore, we introduce a class of programs called {\em flow-bounded} for which the problem of deciding whether there exists a $k>0$ for which the program is $k$-synchronizable, is decidable.
\end{abstract}



\section{Introduction}

Communication with asynchronous message passing is widely used in concurrent and distributed programs implementing various types of systems such as cache coherence protocols, communication protocols, protocols for distributed agreement, web applications, device drivers, etc. 
An asynchronous message passing program is built as a collection of processes running in parallel, communicating asynchronously by sending messages to each other via channels or message buffers. Messages sent to a given process are stored in its entry buffer, waiting for the moment they will be received by the process. In general, sending messages is not blocking for the sender process, which means that the message buffers are supposed to be of unbounded size. 

It is notorious that such programs are hard to get right. Indeed, asynchrony introduces a tremendous amount of new possible interleavings between actions of parallel processes, and makes very hard to apprehend the effect of all of their computations. 
Due to this complexity, expressing and verifying properties such as invariants for such systems is extremely hard. In particular, when buffer are ordered (FIFO buffers), the verification of invariants (or dually of reachability queries) is undecidable even when each of the processes is finite-state~\cite{DBLP:journals/jacm/BrandZ83}.

Therefore, an important issue is the design of verification approaches that avoid considering the full sets of computations 
to draw useful conclusions about the correctness of the considered programs. Several such approaches have been proposed including partial-order techniques, bounded analysis techniques, etc., e.g., \cite{DBLP:journals/tcs/BasuB16,DBLP:conf/oopsla/Desai0M14,DBLP:conf/tacas/BouajjaniE12,DBLP:conf/tacas/TorreMP08,DBLP:conf/popl/FlanaganG05}. Due to the hardness of the problem and its undecidability, these techniques have different limitations: either applicable only when buffers are bounded (e.g., partial-order techniques), or limited in scope, or do not provide any guarantees of termination or insight about completeness of the analysis.

In this paper, we propose a new approach for the analysis and verification of asynchronous message-passing programs with unbounded FIFO buffers, which provides a decision procedure for checking state reachability for a wide class of programs, and which is also applicable for bounded-analysis in the general case. 

We first define a bounding concept for prioritizing the enumeration of program behaviors. 
Our 
intuition comes from the conviction we have that the behaviors of well designed systems can be seen as successions of {\em bounded interaction phases}, each of them being a sequence of send actions (by different processes), followed by a sequence of receive actions (again by different processes) corresponding to send actions belonging to the same interaction phase. For instance, interaction phases corresponding to {\em rendezvous communications} are formed of a single send action followed immediately by its corresponding receive. More complex interactions are the result of exchanges of messages between processes. For instance two processes can send messages to each other, and therefore their interaction starts with two send actions (in any order), followed by the two corresponding receive actions (again in any order). This exchange schema can 
be generalized to any number of processes. 
We say that an interaction phase is {\em $k$-bounded}, for a given $k > 0$, if its number of send actions is less than or equal to $k$. For instance rendezvous interactions are precisely 1-bounded phases.  In general, we call {\em $k$-exchange} any $k$-bounded interaction phase. 
Given $k > 0$, we consider that a computation is {\em $k$-synchronous} if it is a succession of $k$-exchanges.
It can be seen that, in $k$-synchronous computations the sum of the sizes of all messages buffers is bounded by $k$. However, as it will be explained later, boundedness of the messages buffers does not guarantee that there is a $k$ such that all computations are $k$-synchronous. 

Then, we introduce a new bounded analysis which 
for a given $k$, considers only computations that are {\em equivalent} to $k$-synchronous computations. The equivalence relation we consider on computations is based on a notion of {\em trace} corresponding to a {\em happens-before} relation that captures the program order (the order of actions in the code of a process) and the precedence order between send actions and their corresponding receive actions. Two computations are considered to be equivalent if they have the same trace, i.e., they differ only in the order of causally independent actions. 
We show that this analysis is PSPACE-complete.
%

An important feature of our bounding concept is that it is possible to decide its completeness: For any given $k$, it is possible to decide whether every computation of the program (under the asynchronous semantics) is equivalent to (i.e., has the same trace as) a $k$-synchronous computation of that program. 
When this holds, we say that the program is {\em $k$-synchronizable}~\footnote{A different notion of synchronizability has been defined in~\cite{DBLP:journals/tcs/BasuB16} (see Section~\ref{sec:related}).}. Knowing that a program is $k$-synchronizable allows to conclude that an invariant holds for all computations of the program if no invariant violations have been found by its $k$-bounded exchange analysis. Notice that $k$-synchronizability of a program {\em does not} imply that all its behaviours use bounded buffers.
Consider for instance a program with two processes, a producer that consists of a loop of sends, and a consumer that consists of a loop of receives. Although there are computations with arbitrarily large configurations of the entry buffer of the consumer, the program is 1-synchronous because all its computations are equivalent to the computation where each message sent by the producer is immediately received by the consumer. 

Importantly, we show that checking $k$-synchronizability of a program can be reduced in linear time to checking state reachability under the $k$-synchronous semantics (i.e., without considering all the program computations), which implies that checking $k$-synchronizability is PSPACE-complete.
Thus, for $k$-synchronizable programs, it is possible to decide invariant properties without dealing with unbounded message buffers, and the overall complexity in this case is PSPACE. 

Then, a method for verifying asynchronous message passing programs can be defined, 
based on iterating $k$-bounded analyses with increasing value of $k$, starting from $k=1$. If for some $k$, a violation (i.e., reachability of an error state) is detected, then the iteration stops and the conclusion is that the program is not correct. On the other hand, if for some $k$, the program is shown to be $k$-synchronizable and no violations have been found, then again the iteration terminates and the conclusion is that the program is correct. 

However, it might be the case that the program is {\em not} $k$-synchronizable for any $k$. In this case, if the program is correct then the iteration above will not terminate. Thus, an important issue is to determine whether a program is {\em synchronizable}, i.e., {\em there exists a $k$ such that the program is $k$-synchronizable}. This problem is hard, and we believe that it is undecidable, but we do not have a formal proof. However, we are able to define a significant class of programs, including most examples in practice, for which this problem is decidable.  

We have confronted our theory to a set of nontrivial examples. Some of these programs are given as motivating examples in the next section. 
All examples we have found are actually synchronizable (even if all of them are not flow-bounded), which confirms our conviction that non-synchronizability should correspond to an ill-designed system (and therefore it should be reported as an anomaly). Therefore, our approach always terminates and is complete for these systems.

\section{Motivating examples}
\label{sec:motivation}

%
%
%
%
%
%

We provide in this section examples illustrating the relevance and the applicability of our approach. 
Figure \ref{fig:commit} shows a {\em commit protocol} allowing a client to update a memory that is replicated in two nodes. The access to these nodes is controlled by a manager. Figure \ref{fig:commit-exec} shows an execution of this protocol. This system is 1-synchronizable, i.e., every execution of this system is equivalent to one where only rendezvous communication is used. Intuitively, this holds because mutually interacting components are never in the situation where messages sent from one side to the other one are crossing messages sent in the other direction (i.e., the components are "talking" to each other at the same time). For instance, the execution in \ref{fig:commit-exec} is 1-synchronizable because its \emph{conflict graph} (shown in the same figure) is acyclic. Nodes in the conflict graph are matching send-receive pairs (numbered from 1 to 6 in the figure), and edges correspond to the program order between actions in these pairs. The conflict graph being acyclic means that matching pairs of send-receive actions are ``serializable'', which implies that it is equivalent to an execution where every send is immediately followed by the matching receive (as in rendezvous communication).

\begin{figure}[t]
\begin{center}
\includegraphics[width=8.5cm]{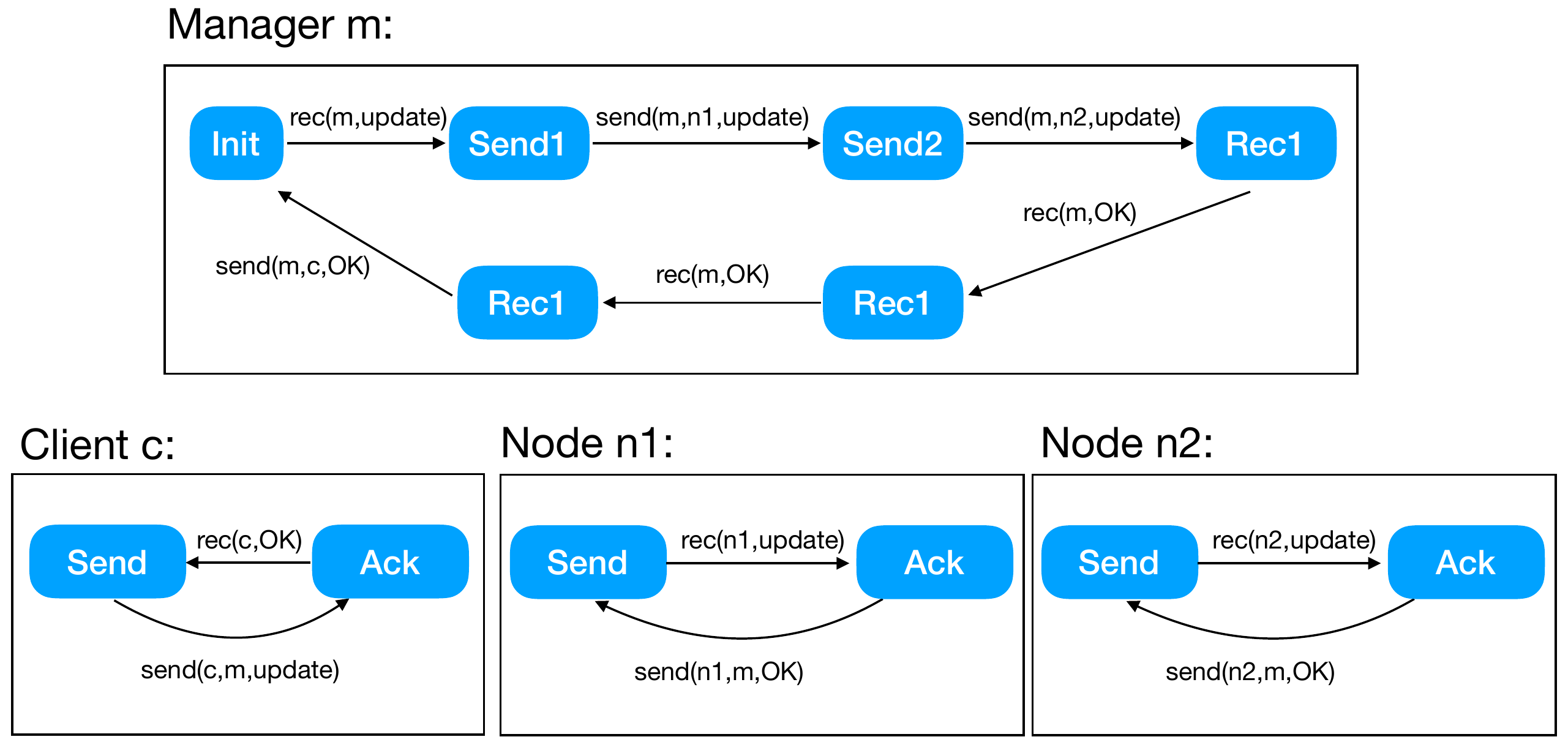}
\end{center}
\vspace{-5.5mm}
\caption{A distributed commit protocol. Each process is defined as a labeled transition system. Transitions are labeled by send and receive actions, e.g., $\senda{\sf{c},\sf{m},\sf{update}}$ is a send from the client $\sf{c}$ to the manager $\sf{m}$ with payload $\sf{update}$. Similarly, $\reca{c,\sf{OK}}$ denotes process $\sf{c}$ receiving a message 
$\sf{OK}$.}
\label{fig:commit}
\vspace{-3.5mm}
\end{figure}

\begin{figure}[t]
\begin{center}
\includegraphics[width=7cm]{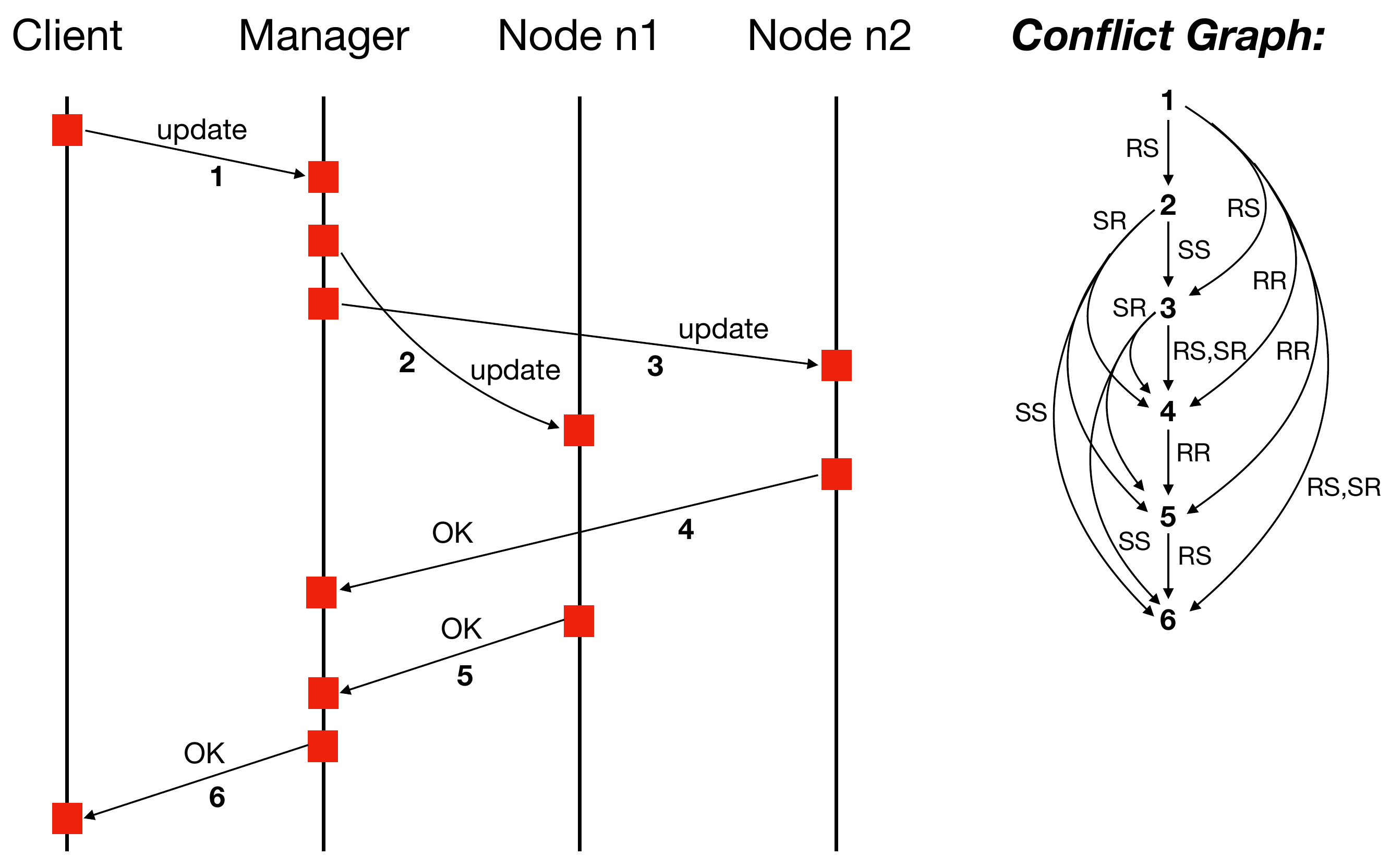}
\end{center}
\vspace{-5mm}
\caption{An execution of the distributed commit protocol and its conflict graph.}
\label{fig:commit-exec}
\vspace{-7mm}
\end{figure}

Although the message buffers are bounded in all the computations of the commit protocol, this is not true for every 1-synchronizable system.
There are asynchronous computations where buffers have an arbitrarily big size, which are equivalent to synchronous computations. This is illustrated for instance by a (family of) computations shown in Figure~\ref{fig:elevator-exec1} of the elevator system shown in Figure \ref{fig:elevator} (a simplified version of the system described in~\cite{DBLP:conf/pldi/DesaiGJQRZ13}).  In this execution, the user keeps sending requests for closing the door, which generates an unbounded sequence of messages in the entry buffer of the elevator process. However, these computations are synchronizable since they are equivalent to a synchronous computation where the elevator receives immediately every message sent by the user. This is witnessed by the acyclicity of the conflict graph of this computation (shown on the right of the same figure). It can be checked that the elevator system shown in Figure \ref{fig:elevator} is a 1-synchronous system (without the dashed edge). 

Consider now a slightly different version of the elevator system where the transition from {\sf Stopping2} to {\sf Opening2} is moved to target {\sf Opening1} instead of {\sf Opening2} (see the dashed transition in Figure \ref{fig:elevator}). It can be seen that this version has the same state space as the previous one. Indeed, moving that transition from {\sf Stopping2} to {\sf Opening1} gives the possibility to {\sf Elevator} to send a message open to {\sf Door}, but the latter can only be between {\sf StopDoor} and {\sf ResetDoor} at this point, and therefore it can (maybe after sending {\sf doorStoped} and {\sf doorOpened}) receive at state {\sf ResetDoor} the message {\sf open} and stay in the same state. However, this version of the system is not 1-synchronizable as it is shown in Figure \ref{fig:elevator-exec2}: Suppose that {\sf Door} is at state {\sf StopDoor}, and that {\sf Elevator} is at state {\sf Stopping2}. Then, {\sf Door} can send a message {\sf doorStoped} and move to the state {\sf OpenDoor}. Next, {\sf Elevator} can receive that message and move to state {\sf Opening1}. At this point, {\sf Elevator} and {\sf Door} can only exchange messages: message {\sf doorOpened} from {\sf Door} to {\sf Elevator} and message {\sf open} from {\sf Elevator} to {\sf Door}. The conflict graph of this execution, shown on the right of Figure \ref{fig:elevator-exec2}, contains a cycle of size 2 between the two matching pairs of send-receive actions involved in the exchange interaction. 

\begin{figure}[t]
\includegraphics[width=12cm]{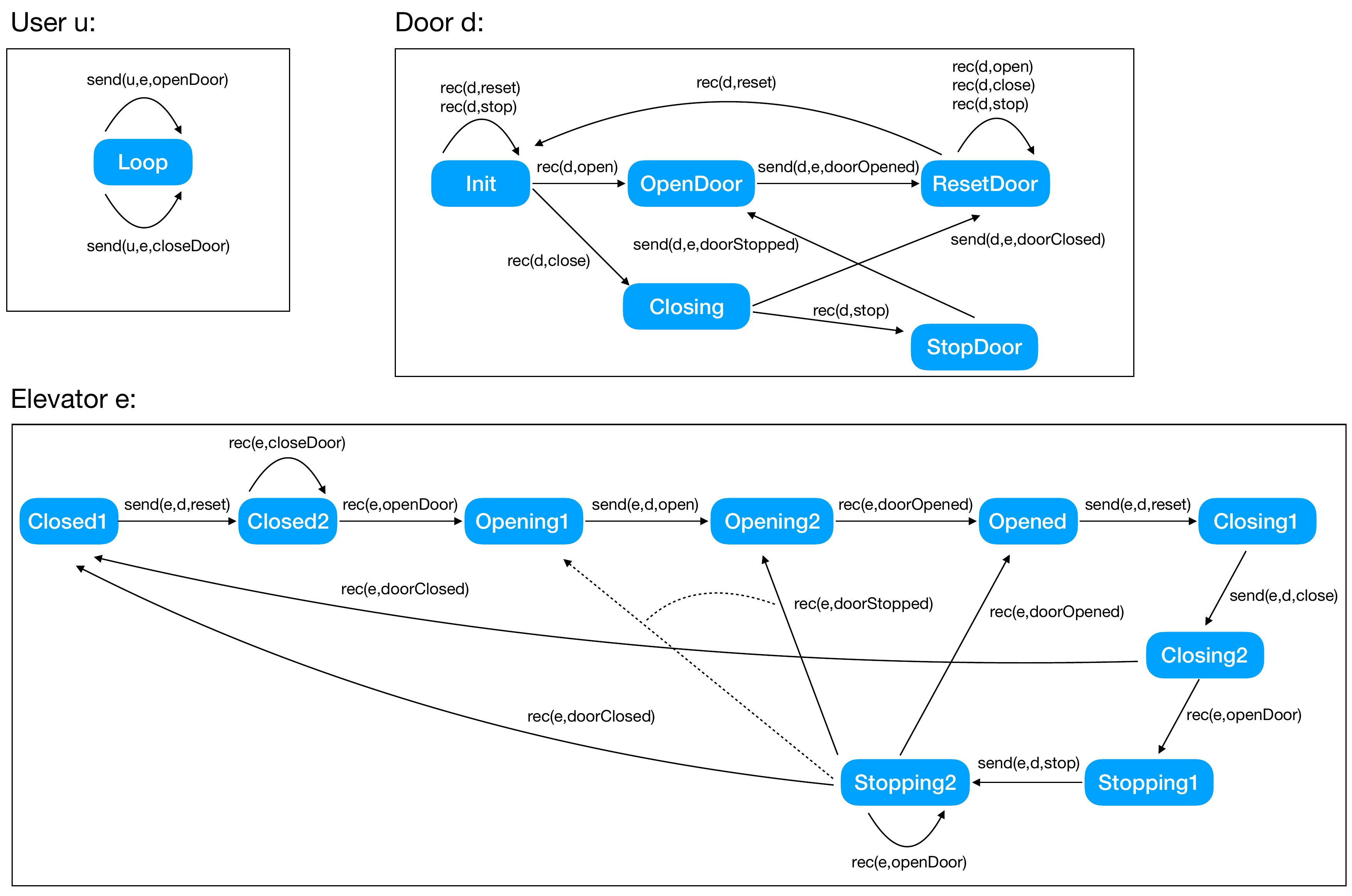}
\vspace{-3mm}
\caption{The Elevator example}
\label{fig:elevator}
\vspace{-5mm}
\end{figure}

\begin{figure}[t]
\begin{subfigure}[t]{6cm}
\includegraphics[width=6cm]{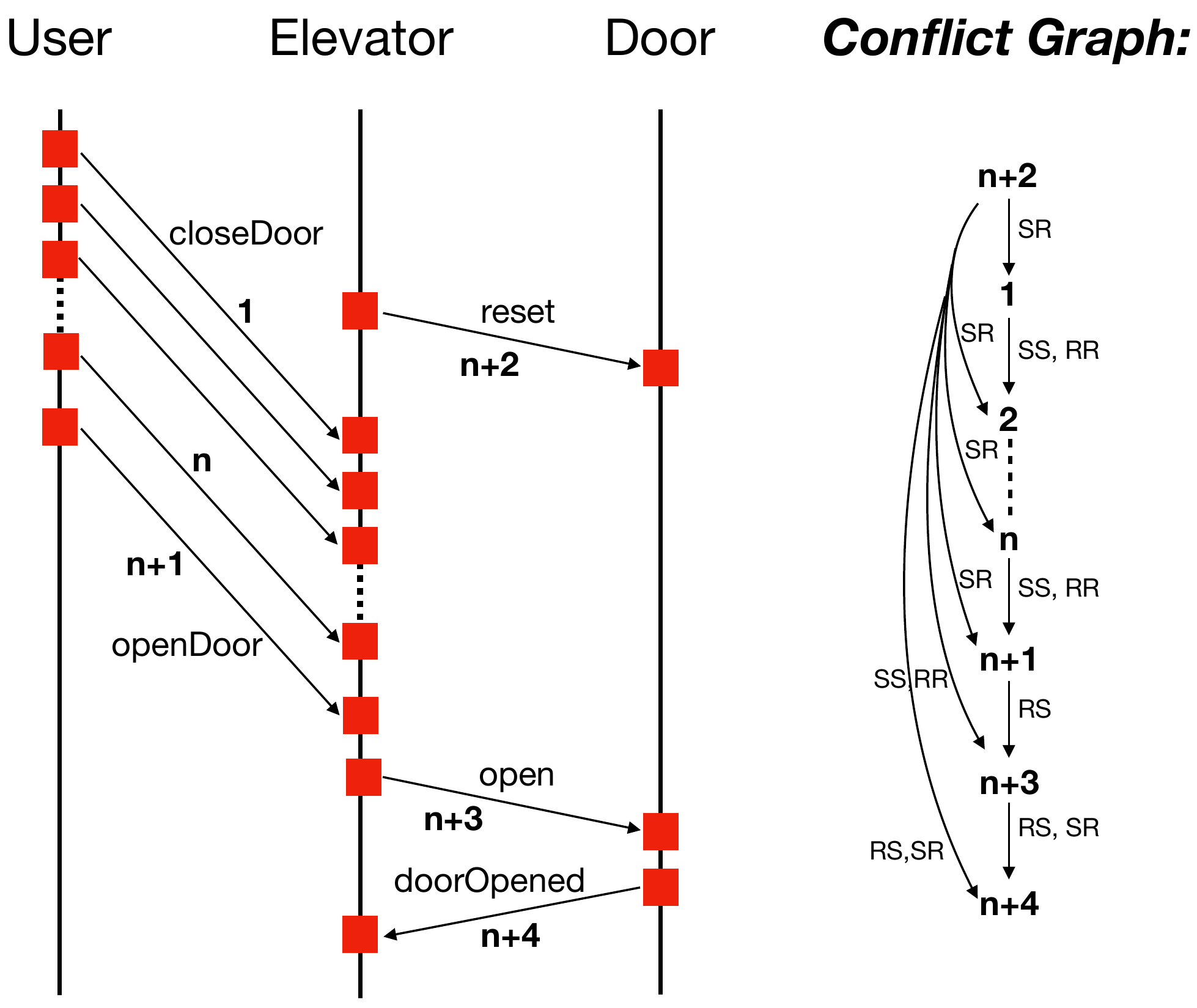}
\vspace{-3mm}
\caption{A $1$-synchronizable execution.}
\label{fig:elevator-exec1}
\end{subfigure}
\hspace{1cm}
\begin{subfigure}[t]{5cm}
\includegraphics[width=5cm]{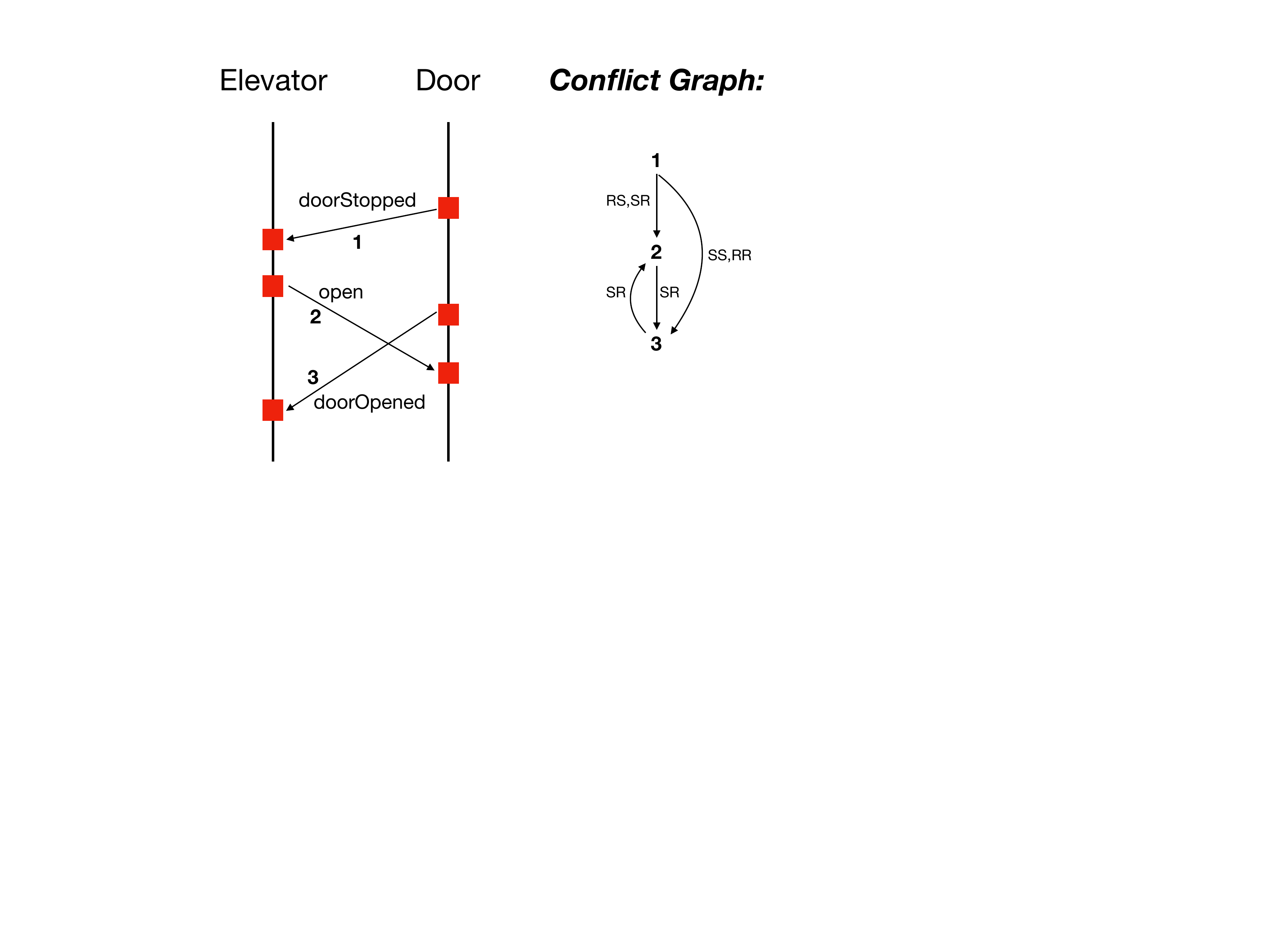}
\vspace{-3mm}
\caption{A computation with a 2-exchange.}
\label{fig:elevator-exec2}
\end{subfigure}
\vspace{-3mm}
\caption{Executions of the Elevator.}
\label{fig:elevator-exec}
\vspace{-6mm}
\end{figure}

\section{Message passing systems}\label{sec:prelims}


We define a message passing system as the composition of a set of processes that exchange messages, which
can be stored in FIFO buffers before being received. Each process is described as a state
machine that evolves by executing send or receive actions.
An execution of such a system can be represented abstractly
using a partially-ordered set of events, called a \emph{trace}. The partial order in a trace represents 
the causal relation between events. We show that these systems satisfy \emph{causal delivery}, i.e., 
the order in which messages are received by a process is
consistent with the causal relation between the corresponding sendings.



We fix sets $\<Pids>$ and $\<Val>$ of process ids and message payloads, 
and sets 
$S = \set{\senda{p,q,v}: p,q\in\<Pids>, v\in \<Val>}$ and $R = \set{\reca{q,v}: q\in\<Pids>, v\in \<Val>}$
of \emph{send actions} and \emph{receive actions}. 
Each send $\senda{p,q,v}$ combines two process ids $p, q$ denoting 
the sender and the receiver of the message, respectively, and a message payload $v$. Receive actions
specify the process $q$ receiving the message, and the message payload $v$. 
The process executing an action $a\in S\cup R$ is denoted $\<proc>(a)$, i.e.,
$\<proc>(a)=p$ for all $a=\senda{p,q,v}$ or $a=\reca{p,v}$,
and the destination $q$ of a send $s=\senda{p,q,v}\in S$ is denoted $\<dest>(s)$.
The set of send, resp., receive, actions $a$ of process $p$, i.e., with $\<proc>(a)=p$, is denoted by $S_p$, resp., $R_p$.


%

A \emph{message passing system} is a tuple $\mathcal{S}=((\<Lsts>_p,\delta_p,l_p^0)\mid p\in\<Pids>)$ 
where $\<Lsts>_p$ is the set of local states of process $p$,
$\delta_p\subseteq \<Lsts>\times (S_p\cup R_p)\times \<Lsts>$ is a transition relation describing the 
evolution of process $p$, and $l^0_p$ is the initial state of process $p$. Examples of message passing systems can be found in Figure~\ref{fig:commit} and Figure~\ref{fig:elevator}.




We fix an arbitrary set $\<Mids>$ of message identifiers, and the sets 
$S_{id} = \set{s_i: s\in S, i\in \<Mids>}$ and $R_{id} = \set{r_i: r\in R, i\in \<Mids>}$
of indexed actions.
Message identifiers are used to pair send and receive actions.
We denote the message id of an indexed send/receive action $a$ by $\<msg>(a)$.
Indexed send and receive actions $s\in S_{id}$ and $r\in R_{id}$ are \emph{matching}, 
written $s\match r$, when $\<msg>(s)=\<msg>(r)$.

An execution of a system $\mathcal{S}$ under the asynchronous semantics is a sequence of indexed actions that is obtained as an interleaving of processes in $\mathcal{S}$ (see Appendix~\ref{asec:semantics} for a formal definition). Every send action enqueues the message into the destination's buffer, and every receive action dequeues a message from the corresponding buffer.
Let $\asynchExec{\mathcal{S}}$ denote the set of these executions.
Given an execution $e$, a send action $s$ in $e$ is called an \emph{unmatched send} when $e$ contains no receive action $r$ such that $s\match r$. An execution $e$ is called \emph{matched} when it contains no unmatched send.


\smallskip
\noindent
{\bf Traces.}
Executions are represented using traces which are sets of indexed actions together with a \emph{program order} relating every two actions of the same process and a \emph{source} relation relating a send with the matching receive (if any).

A trace is a tuple $t=(A,po,src)$ where $A\subseteq S_{id}\cup R_{id}$, $po\subseteq A^2$ defines a total order between actions of the same process, i.e., for every $p\in\<Pids>$, $po\cap \{a: a\in A\mbox{ and }\<proc>(a)=p\}^2$ is a total order, and $src\subseteq S_{id}\times R_{id}$ is a relation s.t. $src(a,a')$ iff $a\match a'$.
The \emph{trace} $tr(e)$ of an execution $e$ is the tuple $(A,po,src)$ where $A$ is the set of all actions in $e$, $po(a,a')$ iff $\<proc>(a)=\<proc>(a')$ and $a$ occurs before $a'$ in $e$, and $src(a,a')$ iff $a\match a'$. Examples of traces can be found in Figure~\ref{fig:commit-exec} and Figure~\ref{fig:elevator-exec}.
The union of $po$ and $src$ is acyclic.
%
%
%
Let $\asynchTr{\mathcal{S}}=\set{tr(e):e\in \asynchExec{\mathcal{S}}}$ be the set of traces of $\mathcal{S}$ under the asynchronous semantics.

Traces abstract away the order of non-causally related actions, e.g., two sends of different processes that could be executed in any order. 
%
Two executions have the same trace when they only differ in the order between such actions.
Formally, given an execution $e=e_1\cdot a\cdot a'\cdot e_2$ with $tr(e)=(A,po,src)$, where $e_1, e_2\in (S_{id}\cup R_{id})^*$ and $a,a'\in S_{id}\cup R_{id}$, we say that $e'=e_1\cdot a'\cdot a\cdot e_2$ is derived from $e$ by a \emph{valid swap} iff $(a,a')\not\in po\cup src$. A permutation $e'$ of an execution $e$ is \emph{conflict-preserving} when $e'$ can be derived from $e$ through a sequence of valid swaps. 
For simplicity, whenever we use the term permutation we mean conflict-preserving permutation.
For instance, a permutation of 
$\send{1}{p_1,q,\_}\ 
\send{2}{p_2,q,\_}\ 
\rec{1}{q,\_}\ 
\rec{2}{q,\_}$
is 
$\send{1}{p_1,q,\_}\ 
\rec{1}{q,\_}\ 
\send{2}{p_2,q,\_}\ 
\rec{2}{q,\_}$
and a permutation of the execution
$\send{1}{p_1,q_1,\_}\ 
\send{2}{p_2,q_2,\_}\ 
\rec{2}{q_2,\_}\ 
\rec{1}{q_1,\_} $
is
$\send{1}{p_1,q_1,\_}\ $
$\rec{1}{q_1,\_}\ 
\send{2}{p_2,q_2,\_}\ 
\rec{2}{q_2,\_}$.


A direct consequence of the definitions is that the set of executions having the same trace are permutations of one another. Also, a system $\mathcal{S}$ cannot distinguish between permutations or equivalently, executions having the same trace.



%

\smallskip
\noindent
{\bf Causal Delivery.}
The asynchronous semantics ensures a property known as \emph{causal delivery}, which intuitively, says that the order in which messages are received by a process $q$ is consistent with the ``causal'' relation between them. Two messages are causally related if for instance, they were sent by the same process $p$ or one of the messages was sent by a process $p$ after the other one was received by the same process $p$. This property is ensured by the fact that the message buffers have a FIFO semantics and a sent message is instantaneously enqueued in the destination's buffer. For instance, the trace (execution) on the left of Figure~\ref{fig:ex-causal-delivery} satisfies causal delivery. In particular, the messages $v1$ and $v3$ are causally related, and they are received in the same order by $q2$. On the right of Figure~\ref{fig:ex-causal-delivery}, we give a trace where the messages $v_1$ and $v_3$ are causally related, but received in a different order by $q2$, thus violating causal delivery. 
This trace is not valid because the message $v1$ would be enqueued in the buffer of $q2$ before $\senda{p,q1,v2}$ is executed and thus, before $\senda{q1,q2,v3}$ as well.

\begin{figure}[t]
\includegraphics[width=9cm]{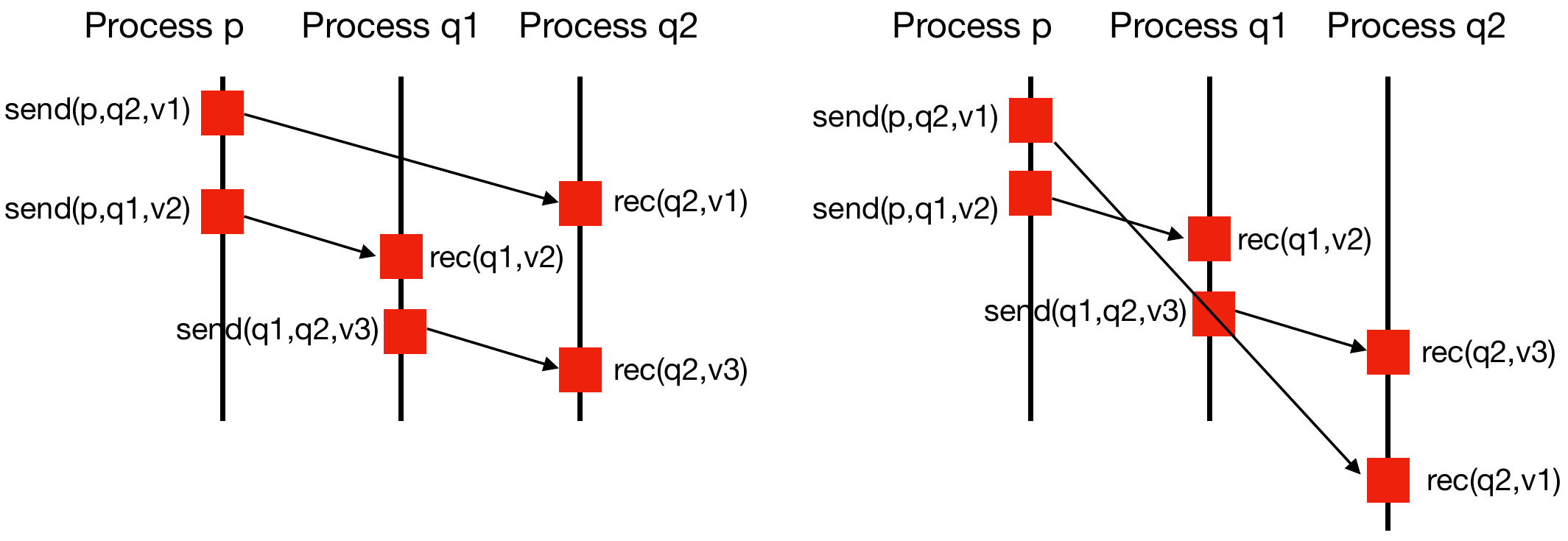}
\vspace{-3mm}
\caption{A trace satisfying causal delivery (on the left) and a trace violating causal delivery (on the right).}
\label{fig:ex-causal-delivery}
\vspace{-3mm}
\end{figure}

Formally, for a trace $t=(A,po,src)$, the transitive closure of $po\cup src$, denoted by $\leadsto_t$, is called the \emph{causal relation} of $t$. For instance, for the trace $t$ on the left of Figure~\ref{fig:ex-causal-delivery}, we have that $\senda{p,q2,v1}\leadsto_t \senda{q1,q2,v3}$.
A trace $t$ satisfies \emph{causal delivery} if 
for every two send actions $s_1$ and $s_2$ in $A$, 
\begin{align*}
&(s_1\leadsto_{t} s_2\land \<dest>(s_1)=\<dest>(s_2))\implies (\not\exists r_2\in A.\ s_2\match r_2)\lor \\
&\hspace{5cm}(\exists r_1,r_2\in A.\ s_1\match r_1\land s_2\match r_2\land (r_2,r_1)\not\in po)
\end{align*}

It can be easily proved that every trace $t\in \asynchTr{\mathcal{S}}$ satisfies causal delivery. 


\section{Synchronizability}\label{sec:criterion}

We define a property of message passing systems called \emph{$k$-synchronizability} as the equality between the set of traces generated by the asynchronous semantics and the set of traces generated by a particular semantics called \emph{$k$-synchronous}. 

\begin{figure}[t]
\begin{center}
\includegraphics[width=6cm]{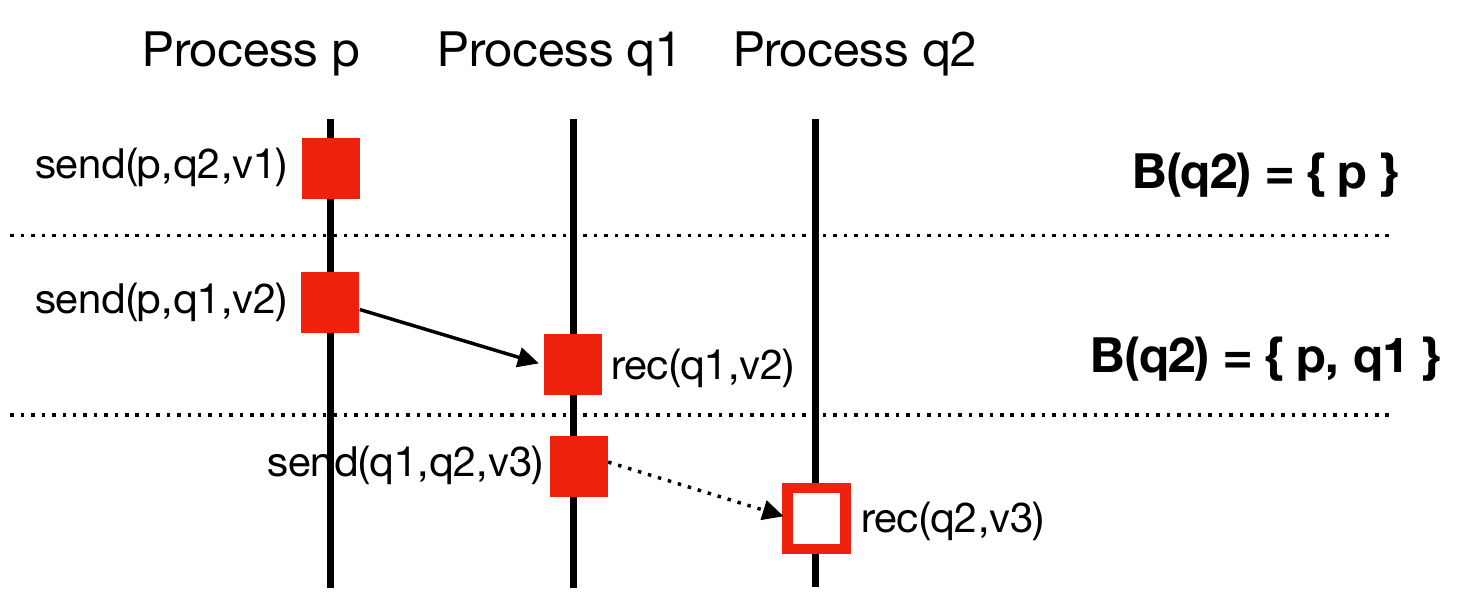}
\end{center}
\vspace{-3mm}
\caption{An execution of the $1$-synchronous semantics.}
\label{fig:ex-blocking}
\vspace{-6mm}
\end{figure}

The $k$-synchronous semantics uses an extended version of the standard rendez-vous primitive where more than one process is allowed to send a message and a process can send multiple messages, but all these messages must be received before being allowed to send more messages. This primitive is called \emph{$k$-exchange} if the number of sent messages is smaller than $k$. For instance, the execution 
$
\send{1}{p_1,q,\_}\ 
\send{2}{p_2,q,\_}\ 
\rec{1}{q,\_}\ 
\rec{2}{q,\_} 
$
is an instance of a $2$-exchange. 
Actually, to ensure that the $k$-synchronous semantics is prefix-closed (if it admits an execution, then it admits all its prefixes), we allow messages to be dropped during a $k$-exchange transition. 
For instance, the prefix of the previous execution without the last receive ($\rec{2}{q,\_}$) is also an instance of a $2$-exchange. The presence of unmatched send actions must be constrained in order to ensure that the set of executions admitted by the $k$-synchronous semantics satisfies causal delivery. Consider for instance the execution in Figure~\ref{fig:ex-blocking} which can be produced by a sequence of $1$-exchanges. The receive action ($\reca{q_2,v_3}$) pictured as an empty box needs to be disabled in order to exclude violations of causal delivery. To this, the semantics tracks for each process $p$ a set of processes $B(p)$ from which it is forbidden to receive messages. Following the sequence of $1$-exchanges in this execution, the unmatched $\senda{p,q2,v1}$ disables any receive by $q2$ of a message sent by $p$ (otherwise, it will be even a violation of the FIFO semantics of $q2$'s buffer). Therefore, the first $1$-exchange results in $B(q2)=\{p\}$. The second $1$-exchange (the message from $p$ to $q1$) forbids $q2$ to receive any message from $q1$. Otherwise, this message will be necessarily causally related to $v1$, and receiving it will lead to a violation of causal delivery. Therefore, when reaching $\senda{q1,q2,v_3}$ the receive $\reca{q_2,v_3}$ is disabled because $q1\in B(q2)$.


\begin{figure} [t]
\footnotesize{
  \centering
  \begin{mathpar}
    \inferrule[$k$-exchange]{
      e\in S_{id}^*\cdot R_{id}^* \\ 
       |e| \leq 2\cdot k\\
      (\vec{l},\vec{\epsilon})\xrightarrow{e} (\vec{l'},\vec{b}),\mbox{ for some $\vec{b}$}\\
            \forall s,r\in e.\ s\match r\implies\<proc>(s)\not\in B(\<dest>(s)) \\
       B'(q) = B(q) \cup \{p: \exists s\in e\cap S_{id}.\ ((\not\exists r\in e.\ s\match r)\land p = \<proc>(s)\land q=\<dest>(s)) \\
        \hspace{3.5cm}\lor (\<proc>(s)\in B(q)\land \<dest>(s) = p)\}
    }{
      (\vec{l},B)
      \xRightarrow{e}{_k}
      (\vec{l'},B')
    }
    
    
  \end{mathpar}
  }
 \vspace{-6mm}
  \caption{The synchronous semantics of a message passing system $\mathcal{S}$. Above, $\vec{\epsilon}$ denotes a vector where all the components are $\epsilon$.
  }
  \label{fig:synch-sem}
\vspace{-5mm}
\end{figure}

Formally, a configuration $c'=(\vec{l},B)$ in the synchronous semantics is a vector $\vec{l}$ of local states together with a function $B:\<Pids>->2^{\<Pids>}$. The transition relation $\Rightarrow_k$ is defined in Figure~\ref{fig:synch-sem}. A \textsc{$k$-exchange} transition corresponds to a sequence of transitions of the asynchronous semantics starting from a configuration with empty buffers. The sequence of transitions is constrained to be a sequence of at most $k$ sends followed by a sequence of receives. The receives are enabled depending on previous unmatched sends as explained above, using the function $B$.
The semantics defined by $\Rightarrow_k$ is called the $k$-synchronous semantics.

Executions and traces are defined as in the case of the asynchronous semantics, using $\Rightarrow_k$ for some fixed $k$ instead of $\rightarrow$. The set of executions, resp., traces, of $\mathcal{S}$ under the $k$-synchronous semantics is denoted by $\synchExec{\mathcal{S}}{k}$, resp., $\synchTr{\mathcal{S}}{k}$. The executions in $\synchExec{\mathcal{S}}{k}$ and the traces in 
$\synchTr{\mathcal{S}}{k}$ are called $k$-synchronous. 

An execution $e$ such that $tr(e)$ is $k$-synchronous is called $k$-synchronizable. We omit $k$ when it is not important. 
The set of executions generated by a system $\mathcal{S}$ under the $k$-synchronous semantics is prefix-closed. Therefore, the set of its $k$-synchronizable executions is prefix-closed as well.
%
Also, $k$-synchronizable and $k$-synchronous executions are undistinguishable up to permutations.


\vspace{-1mm}
\begin{definition}\label{def:synchron}
A message passing system $\mathcal{S}$ is called \emph{$k$-synchronizable} when $\asynchTr{\mathcal{S}}=\synchTr{\mathcal{S}}{k}$.
\vspace{-1mm}
\end{definition}


It can be easily proved that $k$-synchronizable systems reach exactly the same set of local state vectors under the asynchronous and the $k$-synchronous semantics. Therefore, any assertion checking or invariant checking problem for a $k$-synchronizable system $\mathcal{S}$ can be solved by considering the $k$-synchronous semantics instead of the asynchronous one. 
In particular, this implies that such problems are decidable for finite-state $k$-synchronizable systems~\footnote{A system is called \emph{finite-state} when the number of local states of every process is bounded.} 
%
Appendix~\ref{asec:deadlocks} shows that the problem of detecting deadlocks in a $k$-synchronizable system can also be solved on the $k$-synchronous semantics instead of the asynchronous one.

\vspace{-1mm}

\section{Characterizing Synchronous Traces}\label{sec:characterizations}

We give a characterization of the traces generated by the $k$-synchronous semantics that uses a notion of \emph{conflict-graph} similar to the one used in conflict serializability~\cite{journals/jacm/Papadimitriou79b}. The nodes of the conflict graph correspond to pairs of matching actions (a send and a receive) or to unmatched sends, and the edges represent the program order relation between the actions represented by these nodes. 
\begin{wrapfigure}{l}{3cm}
\vspace{-5mm}
\includegraphics[width=3cm]{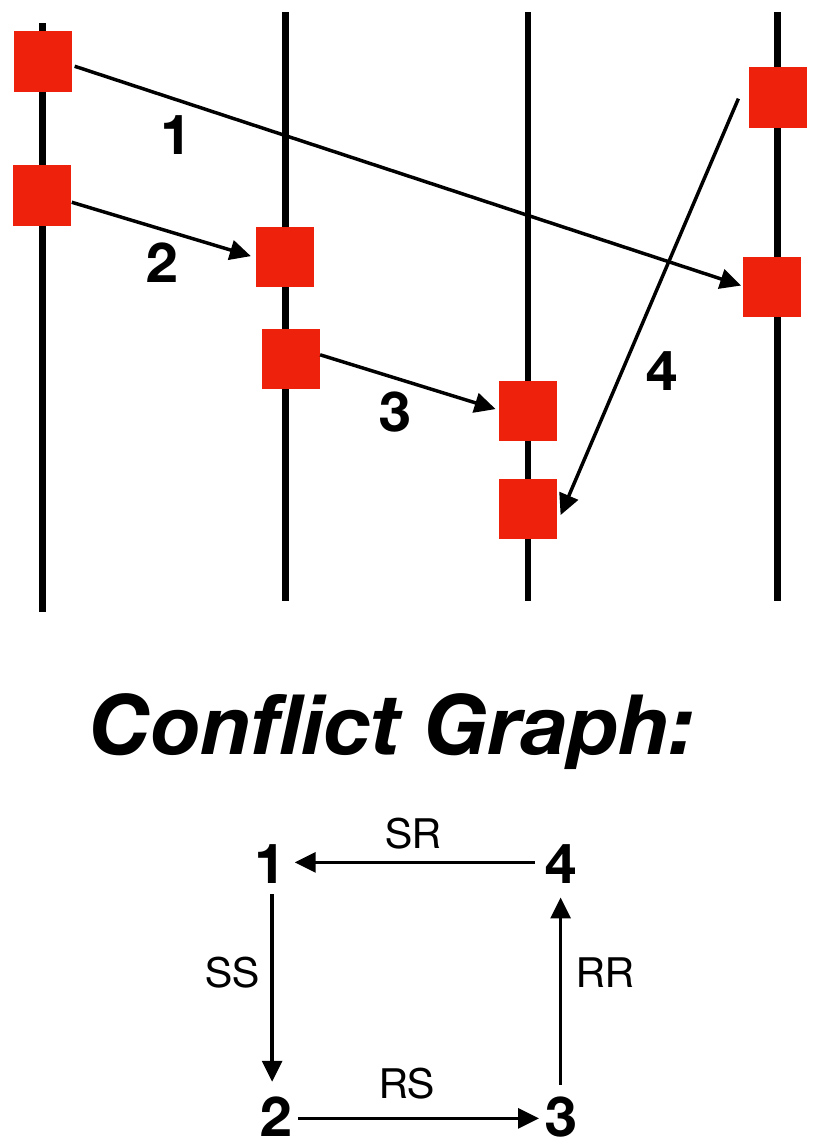}
\caption{ }
\label{fig:ex-rs-cycle}
\vspace{-8mm}
\end{wrapfigure}
For instance, an execution for which the conflict graph of its trace is acyclic, e.g., the execution in Figure~\ref{fig:commit-exec}, is ``equivalent'' to an execution where every receive immediately follows the matching send. 
In general, it is an execution of the $1$-synchronous semantics. For arbitrary values of $k$, the conflict graph may contain cycles, but of a particular form. For instance, traces of the $2$-synchronous semantics may contain a cycle of size 2 like the one in Figure~\ref{fig:elevator-exec}(b). More generally, we show that the conflict graph of a $k$-synchronous trace cannot contain cycles of size strictly bigger than $k$. However, this class of cycles is not sufficient to characterize precisely the $k$-synchronous traces. Consider for instance the trace on the left of Figure~\ref{fig:ex-rs-cycle}. Its conflict-graph contains a cycle of size $4$ (shown on the right), but the trace is not $4$-synchronous. The reason is that the messages tagged by $1$ and $4$ must be sent during the same exchange transition, but receiving message $4$ needs that the message $3$ is sent after $2$ is received. Therefore, it is not possible to schedule all the send actions before all the receives. Such scenarios correspond to cycles in the conflict graph where at least one receive is before a send in the program order. We show that excluding such cycles, in addition to cycles of size strictly bigger than $k$, is a precise characterization of $k$-synchronous traces.

    The \emph{conflict-graph} of a trace $t=(A,po,src)$ is the labeled directed graph $CG_t=\tup{V,E,\ell_E}$ where:
(1) the set of nodes $V$ includes one node for each pair of matching send and receive actions, and each unmatched send action in $t$, and 
(2) the set of edges $E$ is defined by: $(v,v') \in E'$ iff there exist actions $a \in \mathrm{act}(v)$ and $a' \in \mathrm{act}(v')$ such that $(a,a') \in po$ (where $\mathrm{act}(v)$ is the set of actions of trace $t$ corresponding to the graph node $v$). The label of the edge $(v,v')$ records whether $a$ and $a'$ are send or receive actions, i.e., for all $X,Y\in \{S,R\}$, $XY\in \ell(v,v')$ iff $a\in X_{id}$ and $a'\in Y_{id}$.

A direct consequence of previous results on conflict serializability~\cite{journals/jacm/Papadimitriou79b} is that 
a trace is $1$-synchronous whenever its conflict-graph is acyclic.
%
%
A cycle of a conflict graph $CG_t$ is called \emph{bad} when it contains 
an edge labeled by $RS$.
Otherwise, it is called \emph{good}.
The following is a characterization of $k$-synchronous traces (see Appendix~\ref{asec:characterizations} for a proof).


\begin{theorem}\label{lem:cg_k}
A trace $t$ satisfying causal delivery is $k$-synchronous if{f} every cycle in its conflict-graph is good and of size at most $k$.
\end{theorem}

Theorem~\ref{lem:cg_k} can be used to define a runtime monitor checking for $k$-synchronizability. 
The monitor maintains the conflict-graph of the trace produced by the system and checks whether it contains some bad cycle, or a cycle of size bigger than $k$.
While this approach requires dealing with unbounded message buffers, the next section shows that this is not necessary. Synchronizability violations, if any, can be exposed by executing the system under the \emph{synchronous} semantics.

\section{Checking Synchronizability}\label{sec:verif}

We show that checking $k$-synchronizability can be reduced 
to a reachability problem in a system
that executes under the \emph{$k$-synchronous} semantics 
(where message buffers are bounded). We show that every
\emph{borderline} synchronizability violation (for which every strict prefix is synchronizable) of a system $\mathcal{S}$ can be ``simulated''
by the synchronous semantics of a system $\mathcal{S'}$ where the reception of exactly one message is delayed (w.r.t. the synchronous semantics of $\mathcal{S}$).
Then, we give a monitor which observes executions of $\mathcal{S'}$ and identifies synchronizability violations
(there exists a run of this monitor that goes to error whenever such a violation exists). Proofs of the results in this section can be found in Appendix~\ref{sec:verif}.

\subsection{Borderline Synchronizability Violations}\label{ssec:verif1}

For a system $\mathcal{S}$, a violation to $k$-synchronizability $e$ is called \emph{borderline} when every strict prefix of 
$e$ is $k$-synchronizable. Figure~\ref{fig:ex-border-sim}(a) gives an example of a borderline violation to $1$-synchronizability (it is the same execution as in Figure~\ref{fig:elevator-exec}(b)).

We show that every borderline violation $e$ ends with a receive action and this action is included in every cycle of $CG_{tr(e)}$ that is 
bad or exceeds the bound $k$. Given a cycle $c = v,v_1,\ldots,v_n,v$ of a conflict graph $CG_t$, the node $v$ is called a \emph{critical} node of $c$ when $(v,v_1)$ is an $SX$ edge with $X\in \{S,R\}$ 
and $(v_n,v)$ is an $YR$ edge with $Y\in \{S,R\}$.

\begin{figure}[t]
\includegraphics[width=11cm]{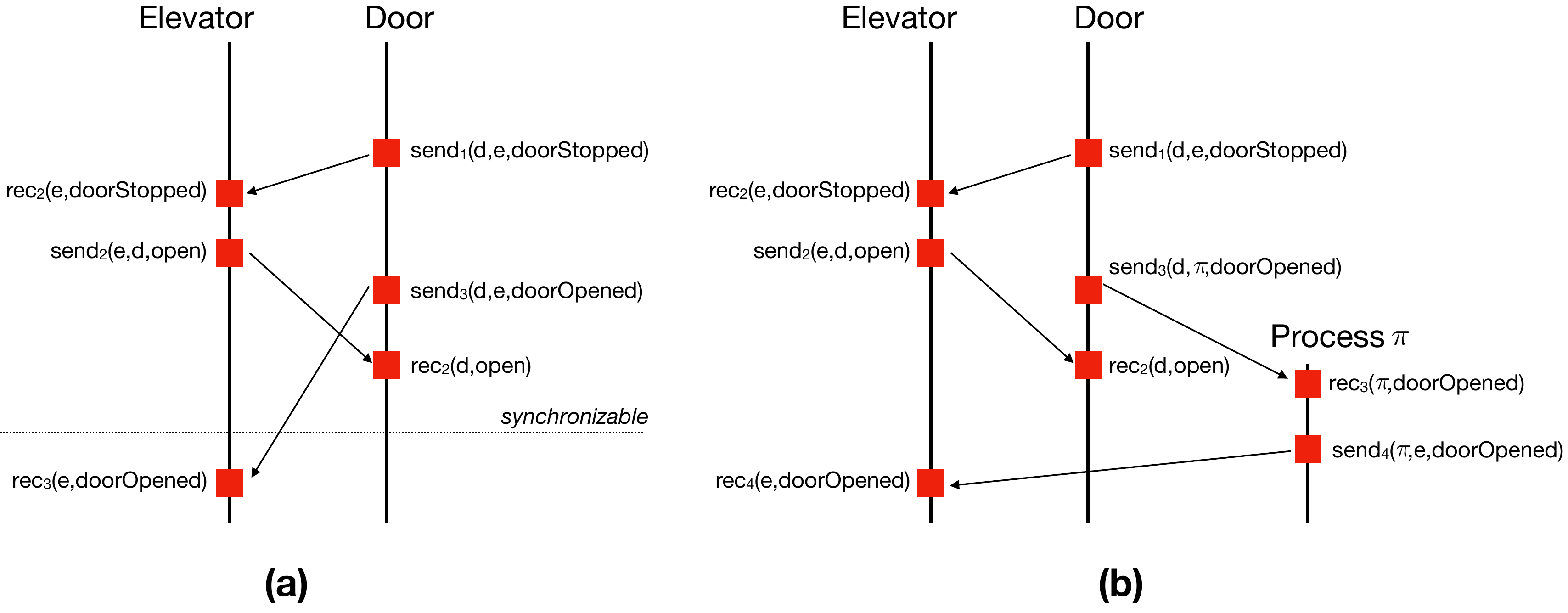}
\caption{A borderline violation to $1$-synchronizability.}
\label{fig:ex-border-sim}
\vspace{-5mm}
\end{figure}

\begin{lemma}
Let $e$ be a borderline violation to $k$-synchronizability of $\mathcal{S}$. Then, $e = e'\cdot r$ for some $e'\in (S_{id}\cup R_{id})^*$ and $r\in R_{id}$.
Moreover, the node $v$ of $CG_{tr(e)}$ representing $r$ (and the corresponding send) is a critical node of every cycle of 
$CG_{tr(e)}$ which is bad or of size bigger than $k$. 
\end{lemma}

\vspace{-1mm}

\subsection{Simulating Borderline Violations on the Synchronous Semantics}\label{ssec:verif2}

Let $\mathcal{S'}$ be a system obtained from $\mathcal{S}$ by ``delaying'' the reception of exactly one message, chosen nondeterministically: $\mathcal{S'}$ contains an additional process $\pi$ and exactly one message sent by a process in $\mathcal{S}$ can be non-deterministically redirected to $\pi$ which sends it to the original destination non-deterministically at a later time.
We show that the synchronous semantics of $\mathcal{S'}$ ``simulates'' a permutation of every borderline violation of 
$\mathcal{S}$. 
Figure~\ref{fig:ex-border-sim}(b) shows the synchronous execution of $\mathcal{S'}$ that corresponds to the borderline violation in Figure~\ref{fig:ex-border-sim}(a). It is essentially the same except for delaying the reception of $\text{doorOpened}$ by sending it to $\pi$ who relays it to the elevator at a later time.


The following result shows that the $k$-synchronous semantics of $\mathcal{S'}$ ``simulates'' all the borderline violations of $\mathcal{S}$, modulo permutations. 

\begin{lemma}
Let $e=e_1\cdot \send{i}{p,q,v}\cdot e_2\cdot \rec{i}{q,v}$ be a borderline violation to $k$-synchronizability of $\mathcal{S}$. Then, $\synchExec{\mathcal{S'}}{k}$ contains an execution $e'$ of the form: 
\begin{align*}
e'=e_1'\cdot \send{i}{p,\pi,(q,v)}\cdot \rec{i}{\pi,(q,v)}\cdot e_2'\cdot \send{j}{\pi,q,v}\cdot \rec{j}{q,v}
\end{align*}
such that $e_1'\cdot \send{i}{p,q,v} \cdot e_2'$ is a permutation of $e_1\cdot \send{i}{p,q,v}\cdot e_2$.
\end{lemma}

Checking $k$-synchronizability for $\mathcal{S}$ on the system $\mathcal{S'}$ would require that every (synchronous) execution of $\mathcal{S'}$ can be transformed to an execution of $\mathcal{S}$ by applying an homomorphism $\sigma$ where the send/receive pair with destination $\pi$ is replaced with the original send action and the send/receive pair initiated by $\pi$ is replaced with the original receive action (all the other actions are left unchanged). However, this is not true in general. For instance, $\mathcal{S'}$ may admit an execution 
\begin{align*}
\send{i}{p,\pi,(q,v)}\cdot \rec{i}{\pi,(q,v)}\cdot \send{i'}{p,q,v'}\cdot \rec{i'}{q,v'}\cdot \send{j}{\pi,q,v}\cdot \rec{j}{q,v}
\end{align*}
where a message sent after the one redirected to $\pi$ is received earlier (and the two messages were sent by the same process $p$). This execution is possible under the $1$-synchronous semantics of $\mathcal{S'}$. Applying the homomorphism $\sigma$, we get the execution 
$
\send{i}{p,q,v)}\cdot \send{i'}{p,q,v'}\cdot \rec{i'}{q,v'}\rec{i}{q,v}
$
which violates causal delivery and it is thus not possible under the asynchronous semantics of $\mathcal{S}$.
Our solution to this problem is to define a monitor $\mathcal{M}_{\mathit{causal}}$ which excludes such executions of $\mathcal{S'}$ when run under the synchronous semantics, i.e., it goes to an error state whenever applying the homomorphism $\sigma$ leads to a violation of causal delivery. This monitor is based on the same principles that we used to exclude violations of causal delivery in the synchronous semantics in the presence of unmatched sends (the component $B$ from a synchronous configuration).

\vspace{-1mm}

\subsection{Detecting Synchronizability Violations}\label{ssec:verif4}

We complete the reduction of checking $k$-synchronizability to a reachability problem under the $k$-synchronous semantics by defining a monitor $\mathcal{M}_{\mathit{viol}}(k)$ which observes executions in $\mathcal{S}_k' \paral\mathcal{M}_{\mathit{causal}}$ and checks whether they represent violations to $k$-synchronizability. We show that there exists a run of $\mathcal{M}_{\mathit{viol}}(k)$ that goes to an error state whenever such a violation exists. 

Essentially, $\mathcal{M}_{\mathit{viol}}(k)$ observes the sequence of $k$-exchanges in an execution and constructs a conflict graph cycle, interpreting the sequence $\send{i}{p,\pi,(q,v)}\rec{i}{\pi,(q,v)}$ as in the original system $\mathcal{S}$, i.e., as $\send{i}{p,q,v}$, and $\send{i}{\pi,q,v}\rec{i}{q,v}$ as $\rec{i}{q,v}$. 
By Lemma~\ref{lem:critical}, every cycle that is a witness for \emph{non} $k$-synchronizability includes the node representing the pair $\send{i}{p,q,v}$, $\rec{i}{q,v}$. Moreover, the successor of this node in the cycle represents an action that is executed by $p$ and the predecessor an action executed by $q$. Therefore, the monitor searches for a conflict-graph path from a node representing an action of $p$ to a node representing an action of $q$. Whenever it founds such a path it goes to an error state.
The set of executions in $\mathcal{S}_k' \paral\mathcal{M}_{\mathit{causal}}$ for which $\mathcal{M}_{\mathit{viol}}(k)$ goes to an error state 
is denoted by $\mathcal{S}_k' \paral\mathcal{M}_{\mathit{causal}}\paral \neg \mathcal{M}_{\mathit{viol}}(k)$.

\begin{theorem}\label{th:main-verif}
For a given $k$, a system $\mathcal{S}$ is $k$-synchronizable iff 
\begin{align*}
\mathcal{S}_k' \paral\mathcal{M}_{\mathit{causal}}\paral \neg \mathcal{M}_{\mathit{viol}}(k)=\emptyset
\end{align*}
\end{theorem}

\section{Decidability results}\label{sec:decidability}

We investigate several decidability and asymptotic complexity questions concerning the synchronous semantics and synchronizability. 
Proofs of the results in this section can be found in Appendix~\ref{asec:decidability}.

Given a system $\mathcal{S}$, an integer $k$, and a local state $l$, \emph{the reachability problem under the $k$-synchronous semantics} asks whether there exists a $k$-synchronous execution of $\mathcal{S}$ reaching a configuration $(\vec{l},B)$ with $l=\vec{l}_p$ for some $p\in\<Pids>$.

\begin{theorem}\label{th:dec1}
For a finite-state system $\mathcal{S}$, the reachability problem under the $k$-synchronous semantics and the problem of checking $k$-synchronizability of $\mathcal{S}$ are decidable and PSPACE-complete.
\end{theorem}

\begin{figure}[t]
\hspace{2cm}
\includegraphics[width=7cm]{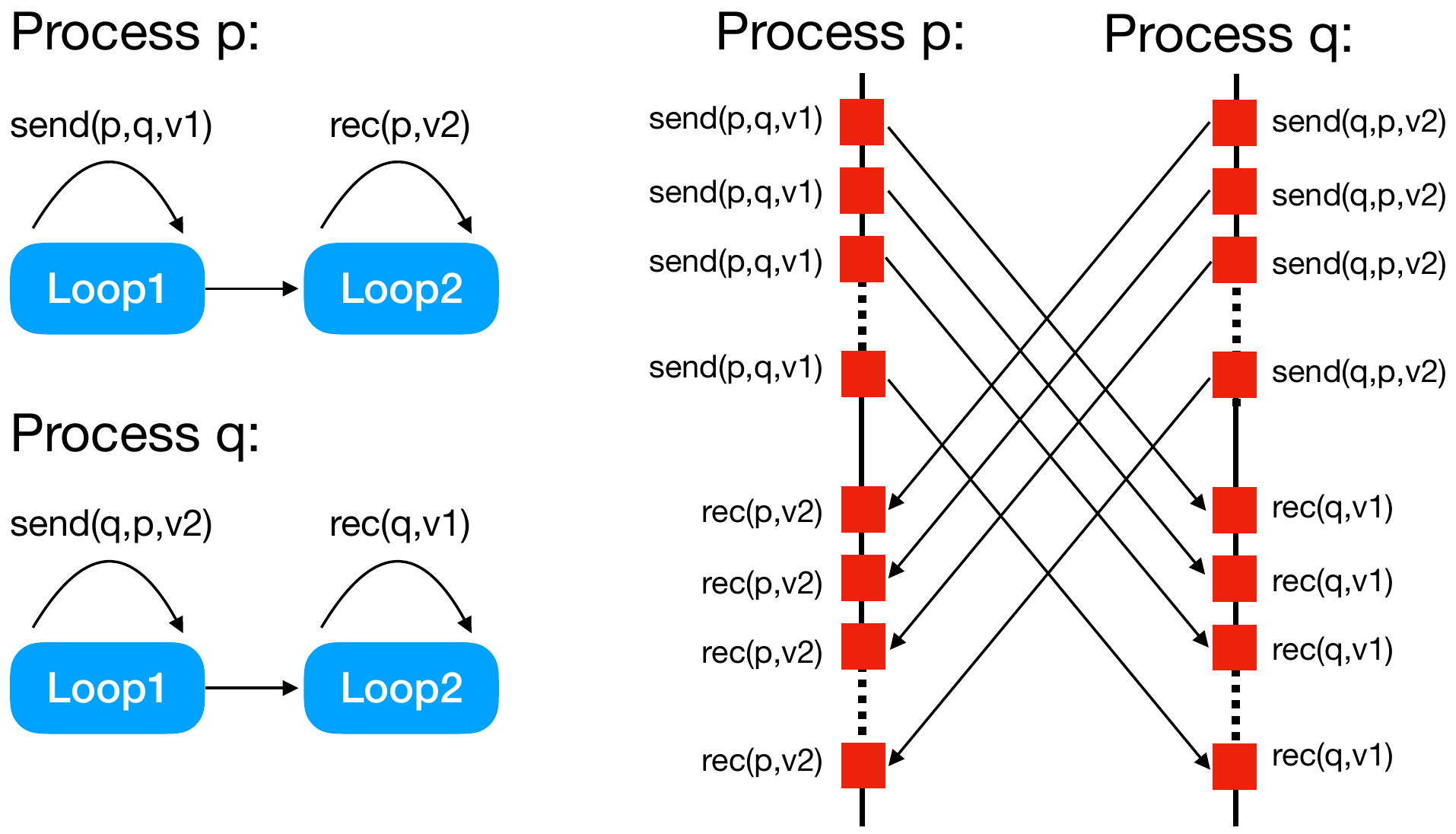}
\caption{An example of a system which is not $k$-synchronizable, for every $k$.}
\label{fig:decid_ex}
\vspace{-5mm}
\end{figure}

We now give a syntactical criterion that imposes an upper bound on the number $k$ for which a system could be $k$-synchronizable.
In general, there are two reasons for which a system is not $k$-synchronizable, for every $k$. It either admits an execution with a bad conflict-graph cycle (e.g., the execution in Figure~\ref{fig:ex-rs-cycle}), or it admits executions with infinitely increasing conflict-graph cycles. If a system admits a bad conflict-graph cycle, then there exists a $k$ for which it can be shown to be non $k$-synchronizable (a coarse upper bound for $k$ is the length of the execution containing this cycle). The second case is exemplified by the system in Figure~\ref{fig:decid_ex}: the two loops in each process allow to create executions with unboundedly many send actions before any receive is enabled. However, the systems we have encountered in practice do not contain such scenarios. 

In fact, the large majority of the processes composing practical systems, e.g., systems developed in the P language~\footnote{Available at \url{https://github.com/p-org}.}, perform a bounded number of consecutive receives, and a bounded number of sends before a receive. If all processes in the system would satisfy this constraint, then there exists a bound $k_s$ on the number of sends that are enabled before a receive, and a bound $k_r$ on the number of receives that are enabled before a send, which would imply that the system is $k$-synchronizable for some $k$ iff it is $k$-synchronizable for some $k\leq (k_s+k_r)\times |\<Pids>|$. However, there exist processes which don't satisfy this constraint, e.g., a consumer in a standard producer-consumer scenario and the process ${\tt Elevator}$ in Figure~\ref{fig:elevator}, which performs an unbounded number of consecutive receive actions. While in the first case, the system would be $1$-synchronous, in the second case, the unbounded number of receives is just an ``optimization'' that doesn't change the set of reachable local state vectors. The self loop where an unbounded number of messages ${\tt closeDoor}$ can be received from the ${\tt User}$ means that all these messages can be ignored since the door of the elevator is anyway closed. This unbounded interaction between ${\tt Elevator}$ and ${\tt User}$ will leave both processes in exactly the same state.
Removing this self loop and considering executions where the ${\tt User}$ sends exactly one message ${\tt closeDoor}$ instead of an unbounded sequence (before a message ${\tt openDoor}$) will allow to discover all the reachable local state vectors. Ignoring the self-loops in ${\tt Door}$ can be motivated in the same way. 

Let $\mathcal{S}=((\<Lsts>_p,\delta_p,l_p^0)\mid p\in\<Pids>)$ be a message passing system. A process $p$ is called \emph{$k$-receive bounded} when it can perform at most $k$ consecutive receives, i.e., for every sequence $w\in (S\cup R)^*$ accepted by the labeled transition system $(\<Lsts>_p,\delta_p,l_p^0)$, there exists no decomposition $w=w_1\cdot w_2\cdot w_3$ where $w_2\in R^*$, and the length of $w_2$ is strictly bigger than $k$. A process $p$ is called \emph{$k$-send bounded} when it can perform at most $k$ consecutive sends before a receive, i.e., for every sequence $w\in (S\cup R)^*$ accepted by the labeled transition system $(\<Lsts>_p,\delta_p,l_p^0)$, there exists no decomposition $w=w_1\cdot w_2\cdot r\cdot  w_3$, where $r\in R$, $w_2\in S^*$, and the length of $w_2$ is strictly bigger than $k$.
For instance, all the processes in the distributed commit protocol in Figure~\ref{fig:commit} are $2$-receive bounded and $2$-send bounded.
The system $\mathcal{S}$ is called \emph{flow-bounded} when there exists a constant $k$ such that every process $p\in\<Pids>$ is $k$-receive bounded and $k$-send bounded.
Note that verifying flow-boundedness is reducible to a reachability problem for a single process and thus, decidable for finite-state processes.

\begin{theorem}
For a flow-bounded system $\mathcal{S}$, the problem of checking if there exists some $k$ such that $\mathcal{S}$ is $k$-synchronizable, is decidable.
\end{theorem}

\section{Related Work}\label{sec:related}

Automatic verification of asynchronous message passing systems is undecidable in general~\cite{DBLP:journals/jacm/BrandZ83}. 
A number of decidable subclasses has been proposed. 
The class of systems, called \emph{synchronizable} as well, in~\cite{DBLP:journals/tcs/BasuB16}, requires that a system generates the same sequence of send actions when executed under the asynchronous semantics as when executed under a synchronous semantics based on rendezvous communication. These systems are all $1$-synchronizable, but the inclusion is strict (the $1$-synchronous semantics allows unmatched sends which cannot be produced by rendezvous communication). The techniques proposed in~\cite{DBLP:journals/tcs/BasuB16} to check that a system is synchronizable according to their definition cannot be extended to $k$-synchronizable systems.
Other classes of systems that are $1$-synchronizable have been proposed in the context of session types, e.g.,~\cite{DBLP:conf/esop/DenielouY12,DBLP:journals/jacm/HondaYC16,DBLP:conf/esop/HondaVK98,DBLP:conf/popl/LangeTY15}. 
Our class of synchronizable systems differs also from other classes of communicating systems that restrict the type of communication, e.g., lossy-communication~\cite{DBLP:journals/iandc/AbdullaJ96}, half-duplex communication~\cite{DBLP:journals/iandc/CeceF05}, or the topology of the interaction, e.g., tree-based communication in particular classes of concurrent push-down systems~\cite{DBLP:conf/tacas/TorreMP08,DBLP:journals/corr/abs-1209-0359}.

The question of deciding if all computations of a communicating system are equivalent (in the language theoretic sense) to computations with bounded buffers has been studied in, e.g., \cite{DBLP:journals/fuin/GenestKM07}, where this problem is proved to be undecidable. The link between that problem and our synchronizability problem is not (yet) clear, mainly because non synchronizable computations may use bounded buffers.

Our work proposes a solution to the question of defining adequate (in terms of coverage and complexity) parametrized bounded analyses for message passing programs, providing the analogous of concepts such as context-bounding or delay-bounding defined for shared-memory concurrent programs. Bounded analyses for concurrent systems
was in fact initiated by the work on bounded-context switch analysis~\cite{DBLP:conf/pldi/QadeerW04,DBLP:conf/tacas/QadeerR05,DBLP:conf/cav/LalR08}. For shared-memory programs, this work has been extended to an unbounded number of threads or larger classes of behaviors, e.g.,~\cite{DBLP:conf/sas/BouajjaniEP11,DBLP:conf/popl/EmmiQR11,DBLP:conf/spin/KiddJV10,DBLP:conf/cav/TorreMP10}. Few bounded analyses incomparable to ours have been proposed for message passing systems, e.g.,~\cite{DBLP:conf/tacas/TorreMP08,DBLP:conf/tacas/BouajjaniE12}. Contrary to our work, all these works on bounded analyses in general do not propose decision procedures for checking if the analysis is complete, i.e., it covers the whole set of reachable states. The only exception that we are aware of is~\cite{DBLP:conf/cav/TorreMP10}, which however concerns shared-memory programs.  

Partial-order reduction techniques, e.g.,~\cite{DBLP:conf/popl/AbdullaAJS14,DBLP:conf/popl/FlanaganG05}, allow to define equivalence classes on behaviors, based on notions of action independence and explore (ideally) only one representative of each class. This has lead to efficient algorithmic techniques for enhanced model-checking of concurrent shared-memory programs that consider only a subset of relevant action interleavings. In the worst case, these techniques will still need to explore all of the interleavings. Moreover, these techniques are not guaranteed to terminate when the buffers are unbounded.

The work in~\cite{DBLP:conf/oopsla/Desai0M14} defines a particular class of schedulers, that roughly, prioritize receive actions over send actions, which is complete in the sense that it allows to construct the whole set of reachable states. Defining an analysis based on this class of schedulers has the same drawback as partial-order reductions, in the worst case, it needs to explore all interleavings, and termination is not guaranteed. 

The notion of conflict-graph is similar to the one used for defining conflict serializability~\cite{journals/jacm/Papadimitriou79b}. However, our algorithms and proof techniques are very different from those used in this context, e.g.,~\cite{DBLP:journals/iandc/AlurMP00,DBLP:conf/esop/BouajjaniEEH13,DBLP:conf/cav/FarzanM08}. Our approach considers several classes of cycles in these graphs and focuses on showing that these cycles can be detected without exploring all the behaviors of a system.

The approach we adopt in this work is related to robustness checking \cite{DBLP:conf/se/BouajjaniDM14,DBLP:conf/esop/BouajjaniEEOT17}. The general paradigm is to decide that a program has the same behaviors under two semantics, one being weaker than the other, by showing a polynomial reduction to a state reachability problem under the stronger semantics, i.e., by avoiding the consideration of the weak semantics that is in general far more complex to deal with than the strong one. For instance, in the case of our work, the class of message passing programs with unbounded FIFO channels is Turing powerful, but still, surprisingly, $k$-synchronizability of these programs is decidable and PSPACE-complete (i.e., as hard as state reachability in programs with bounded channels). However, the results in \cite{DBLP:conf/se/BouajjaniDM14,DBLP:conf/esop/BouajjaniEEOT17} can not be applied to solve the question of synchronizability we consider in this paper; in each of \cite{DBLP:conf/se/BouajjaniDM14}, \cite{DBLP:conf/esop/BouajjaniEEOT17}, and our work, the considered classes of programs and their strong/weak semantics are very different (shared-memory concurrent programs running over a relaxed memory model in \cite{DBLP:conf/se/BouajjaniDM14}, and shared-memory concurrent programs with dynamic asynchronous process creation in \cite{DBLP:conf/esop/BouajjaniEEOT17}), and the corresponding robustness checking algorithms are based on distinct concepts and techniques. 

\section{Experimental Evaluation}

\begin{wrapfigure}{l}{5.7cm}
\vspace{-8mm}
{\scriptsize
\begin{tabular}{|c|c|c|c|c|c|}
\hline
Name & Proc & Loc & Instr & $k$ & Time \\
\hline
Elevator& 3 & 94 & 481 & 2 & 13m  \\
\hline
Two-phase commit &4 & 145 & 381 & 1 & 10m  \\
\hline
Replication Storage & 5 & 243 & 515 & 4 & 15m \\
\hline
German Protocol & 5 & 300 & 637 & 2 & 25m \\
\hline
OSR & 4 & 154 & 464  & 1 & 22m \\
\hline
\end{tabular}
}
\caption{Experimental results.}
\label{fig:experiments}
\vspace{-7mm}
\end{wrapfigure}
As a proof of concept, we have applied our procedure for checking $k$-synchronizability to a set of examples extracted from the distribution of the P language~\footnote{Available at \url{https://github.com/p-org}.}. In the absence of an exhaustive model-checker for this language, we have rewritten these examples in the Promela language and used the Spin model checker~\footnote{Available at \url{http://spinroot.com}} for discharging the reachability queries. For a given Promela program, its $k$-synchronous semantics is implemented as an instrumentation which uses additional boolean variables to enforce that sends and receives interleave in $k$-exchange phases. Then, the monitors defined in Section~\ref{sec:verif} are defined as additional processes which observe the sequence of $k$-exchanges in an execution and update their state accordingly. Finding a conflict-graph cycle which witnesses non $k$-synchronizability corresponds to violating an assertion. 

The experimental data is listed in Figure~\ref{fig:experiments}: Proc, resp., Loc, is the number of processes, resp., the number of lines of code (loc) of the original program, Instr is the number of loc added by the instrumentation, $k$ is the \emph{minimal} integer for which the program is $k$-synchronizable, and Time gives the number of minutes needed for this check. The first three examples are the ones presented in Section~\ref{sec:motivation} and Appendix~\ref{asec:motivation}. The German protocol is a modelization of the cache-coherence protocol with the same name, and OSR is a modelization of a device driver.

\clearpage

\bibliography{references}

\newpage
\appendix
\section{Motivating Examples}\label{asec:motivation}

\begin{figure}[t]
\includegraphics[width=10cm]{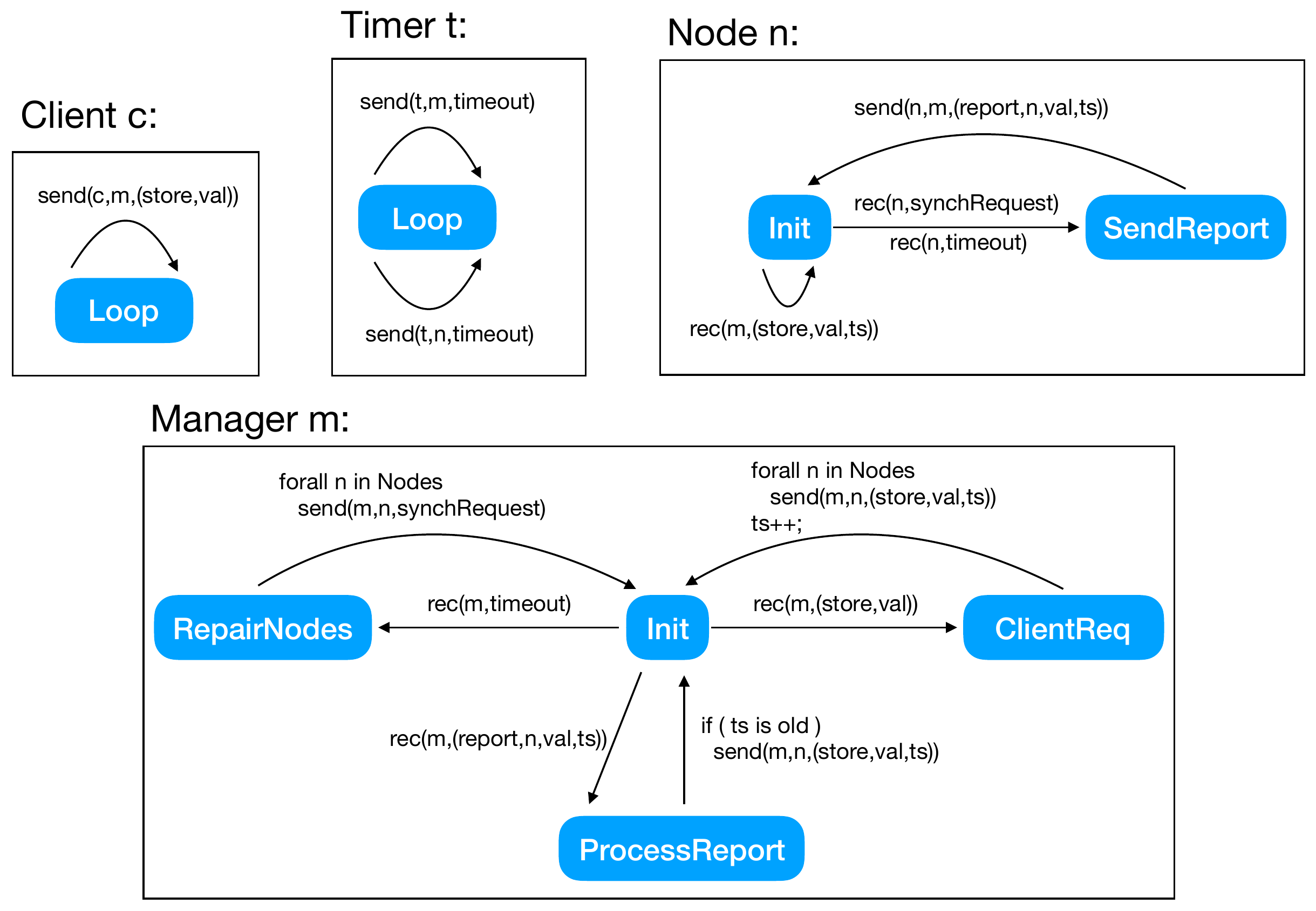}
\caption{A replication storage protocol.}
\label{fig:replication}
\end{figure}

\begin{figure}[t]
\includegraphics[width=7cm]{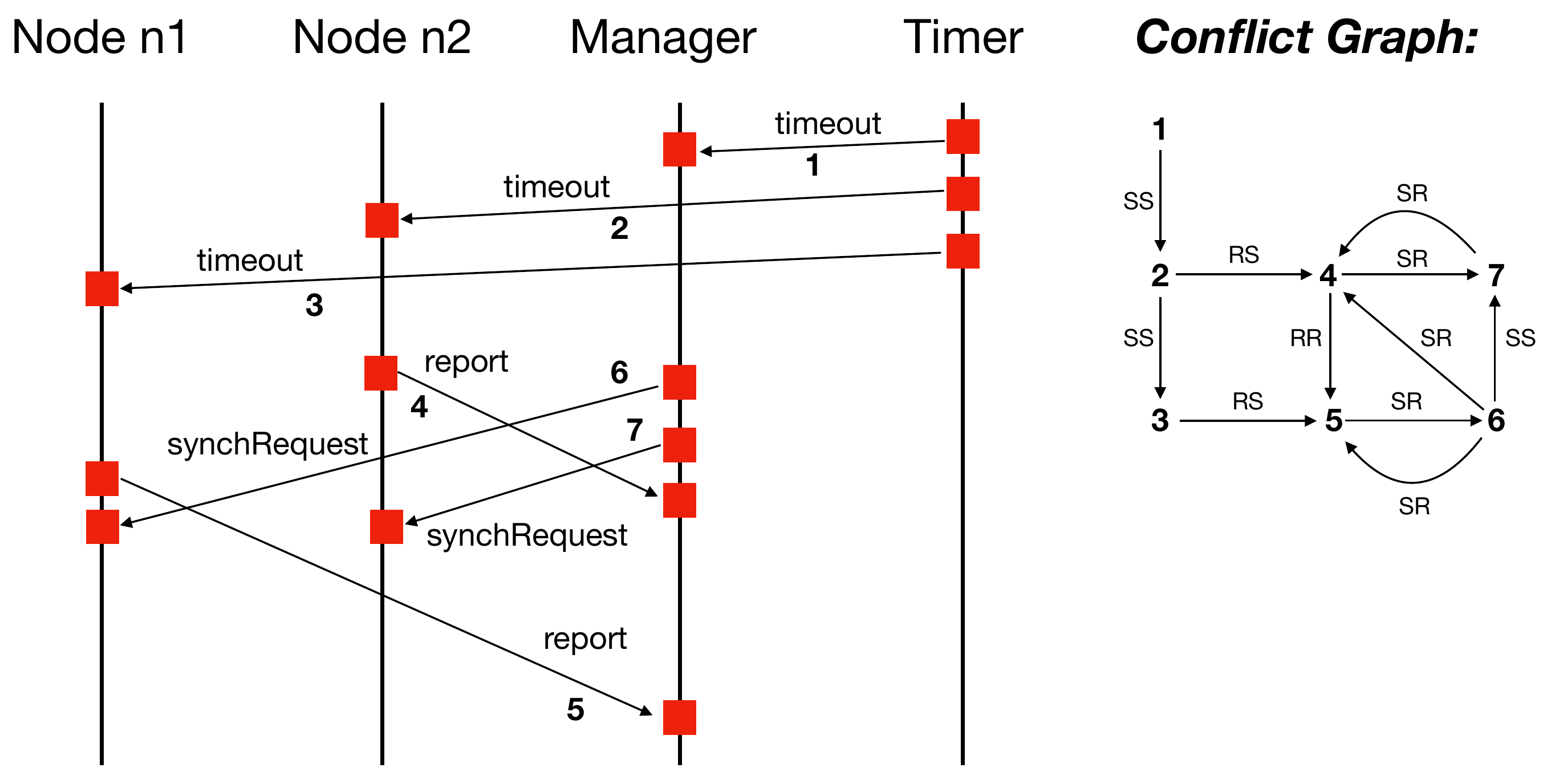}
\caption{An execution of the replication storage protocol and its conflict graph.}
\label{fig:replic-exec}
\end{figure}

We present an example illustrating the fact that exchanges can involve not only 2 processes but several ones, maybe all the processes of a system. Figure \ref{fig:replication} shows a replication storage protocol. A client can send update requests to the manager who is in charge of maintaining several storage replicas on different nodes. When the manager receives an update request from the client, it forwards this message to all the nodes. However, since messages can be delayed, the information in the nodes can be different at various points in time. Then, a mechanism is used to force regularly synchronization between the data versions stored in the different nodes. This mechanism is based on (1) using time-stamps for each message that is sent  by the manager to the nodes, and (2) a timer that triggers synchronizations: the timer can send timeout messages to either the manager or to the nodes. When a node receives such a timeout message, it sends a report to the manager with its current data value, and when the manager receives the timeout message, it sends to all nodes a message requesting a synchronization. When a node receives the synchronization request, it sends to the manager a report with its last value. After each reception of a report from a node, the manager checks if the received value is up-to-date using its time-stamp, and if not, it sends the most recent value to the node. Now, since the nodes and the manager may all receive timeout messages from the timer, nodes can start sending reports to the manager while the latter is already sending them synchronization requests. This leads to an exchange that may involve all the nodes, in addition to the manager. This situation is shown in Figure \ref{fig:replic-exec}. The conflict graph shown on the right of the figure contains a cycle of size 4, which is in this case the number of involved processes (two nodes and one manager), plus 1, which means that the considered computation is not 3-synchronizable. It can be checked that the system is actually 4-synchronizable. 

\section{Asynchronous Semantics of Message Passing Systems}\label{asec:semantics}

Formally, configuration $c=\tup{\vec{l},\vec{b}}$ is a vector  $\vec{l}$  of local states together with a vector $\vec{b}$ of message buffers  that are
represented as sequences of message payloads tagged with unique identifiers. The identifiers are used only for technical reasons, to identify a ``matching'' relation
between sends and receives.  These two vectors are indexed by elements of $\<Pids>$.
For a vector $\vec{x}$, let $\vec{x}_p$ denote the element of $\vec{x}$ of index $p$.
The initial configuration $c_0$ of the system $\mathcal{S}$ is the tuple of initial local states together with empty message buffers, i.e., 
$c_0=\tup{\vec{l},\vec{b}}$ where $\vec{l}_p=l_p^0$ and $\vec{b}=\epsilon$ for each $p\in\<Pids>$.

\begin{figure} [t]
\footnotesize{
  \centering
  \begin{mathpar}
    \inferrule[send]{
      m= (i,v) \\ 
      i\in \<Mids>\mbox{ fresh} \\
      (\vec{l}_p,\senda{p,q,v},l)\in \delta_p 
    }{
      \vec{l},\vec{b}
      \xrightarrow{\send{i}{p,q,v}}
      \vec{l}[\vec{l}_p\gets l],\vec{b}[\vec{b}_q\gets \vec{b}_q\cdot m]
    }
    
    \inferrule[receive]{
      \vec{b}_q = (i,v) \cdot b \\
      (\vec{l}_q,\reca{q,v},l)\in \delta_q
    }{
      \vec{l},\vec{b}
      \xrightarrow{\rec{i}{q,v}}
      \vec{l}[\vec{l}_q\gets l],\vec{b}[\vec{b}_q\gets b]
    }
    
  \end{mathpar}
  }
  \caption{The asynchronous semantics of a message passing system $\mathcal{S}$. For a vector $\vec{x}$, $\vec{x}[\vec{x}_p\gets E]$ is the vector $\vec{y}$ with $\vec{y}_q=\vec{x}_q$, for every $q\neq p$, and $\vec{x}_p=E$. Also, $\cdot$ denotes the concatenation of two sequences.
  }
  \label{fig:asynch-sem}
\end{figure}

The transition relation $\rightarrow$ in Figure~\ref{fig:asynch-sem} is determined by a message passing system $\mathcal{S}$, and maps
a configuration $c_1$ to another configuration $c_2$ and indexed action $a\in S_{id}\cup R_{id}$.
The effect of a \textsc{send} transition is to enqueue the message payload tagged with a fresh identifier to the buffer of the destination, and the effect of a \textsc{receive} transition is to dequeue a message from the local buffer.

An \emph{execution} of a system $\mathcal{S}$ under the asynchronous semantics to configuration ${c}_n$
is a sequence of indexed actions $a_1 \ldots a_n$ such that 
there exists a configuration sequence ${c}_0 {c}_1 \ldots {c}_n$ with
$  {c}_m \xrightarrow{a_{m+1}} {c}_{m+1}$
for all $0 \le m < n$. 
We 
say that ${c}_n$ is reachable in $\mathcal{S}$ under the asynchronous semantics. 
The \emph{reachable local state vectors} of $\mathcal{S}$, denoted by $\asynchSt{\mathcal{S}}$, is the
set of local state vectors in reachable configurations.
The set of executions of $\mathcal{S}$ under the asynchronous semantics is denoted by $\asynchExec{\mathcal{S}}$.
In the following, we don't distinguish between executions obtained by a consistent renaming of the message identifiers.

\section{Detecting deadlocks}\label{asec:deadlocks}

In addition to assertion/invariant checking, we show that the problem of detecting deadlocks in a $k$-synchronizable system can also be solved using the $k$-synchronous semantics instead of the asynchronous one. For a process $p$, a state $l\in \<Lsts>_p$ is called \emph{receiving} when $(l,a,l')\in\delta_p$, for some $l'$, implies that $a\in R_p$. For instance, the state {\tt Init} of the process {\tt Manager} in Figure~\ref{fig:commit} is receiving. The state $l$ is called \emph{final} when there exists no $l'$ and $a$ such that $(l,a,l')\in\delta_p$.

We consider several notions of deadlock: a configuration $c=(\vec{l},\vec{b})$ is called 
\begin{itemize}
	\item \emph{empty-buffer deadlock} when all the buffers are empty, there exists at least one process waiting for a message, and all the other processes are either in a final or receiving state, i.e., $\vec{b}=\vec{\epsilon}$, there exists $p\in\<Pids>$ such that $(\vec{l}_p,r,l')\in\delta_p$ for some $r\in R$, and for all $q\in \<Pids>$, $\vec{l}_q$ is receiving or final,
	\item \emph{orphan message configuration} when there is at least a non-empty buffer and each process is in a final state, i.e., $\vec{b}\neq\vec{\epsilon}$ and for all $p\in \<Pids>$, $\vec{l}_p$ is final,
	\item \emph{unspecified reception} when some process is prevented from receiving any message from its buffer, i.e., there exists $p\in\<Pids>$ such that $\vec{l}_p$ is receiving, and for all $\reca{p,v}\in R$, if $(\vec{l}_p,\reca{p,v},l')\in \delta_p$, for some $l'$, then $\vec{b}_p\not\in v\<Val>^*$.
\end{itemize}

We show that reachability of such configurations in the original asynchronous semantics can be reduced to reachability problems over the synchronous semantics, provided that the system is $k$-synchronizable. The constraints over the buffers of the asynchronous configurations are replaced by constraints over the executions (traces) of the synchronous semantics. For instance, an execution reaching an empty-buffer configuration is ``equivalent'' to a synchronous \emph{matched} execution where every sent message has been received (assuming $k$-synchronizability).

We extend the notion of empty-buffer deadlock to configurations of the synchronous semantics by removing the condition that the buffers are empty.

\begin{theorem}\label{th:deadlock1}
A $k$-synchronizable system $\mathcal{S}$ reaches an empty-buffer deadlock configuration under the asynchronous semantics iff the $k$-synchronous semantics of $\mathcal{S}$ admits a matched execution to an empty-buffer deadlock configuration.
\end{theorem}
\begin{proof}
We prove the only-if direction, the reverse being similar. 
Let $e$ be an execution in $\asynchExec{\mathcal{S}}$ to an empty-buffer deadlock configuration $(\vec{l},\vec{\epsilon})$. Since the buffers are empty, by the definition of the asynchronous semantics, we get that $e$ is matched. By $k$-synchronizability, there exists a permutation $e'$ of $e$ that belongs to $\synchExec{\mathcal{S}}{k}$. Then, by Lemma~\ref{lem:undist}, $e'$ is an execution to a configuration $(\vec{l},B)$, for some $B$, which finishes the proof.
\end{proof}

A configuration $(\vec{l},B)$ of the synchronous semantics is called \emph{final} when every local state $\vec{l}_p$ with $p\in\<Pids>$ is final. The proof of the following result is similar to that of Theorem~\ref{th:deadlock1}, the only addition being that an asynchronous execution to a configuration with non-empty buffers corresponds to a synchronous execution with unmatched send actions (provided that the system is $k$-synchronizable). 

\begin{theorem}
A $k$-synchronizable system $\mathcal{S}$ reaches an orphan message configuration under the asynchronous semantics iff the $k$-synchronous semantics of $\mathcal{S}$ admits an execution containing at least one unmatched send to a final configuration.
\end{theorem}


A local state $l$ of a process $p$ is called \emph{$V$-receiving} when it is receiving and the set of messages that can be received in $l$ is exactly $V$, i.e., for all $v$,  $v\in V$ iff there exists $l'\in \<Lsts>_p$ such that $(l,\reca{p,v},l')\in \delta_p$. A configuration $(\vec{l},B)$ of the synchronous semantics is called \emph{$(p,V)$-receiving} when $\vec{l}_p$ is $V$-receiving. 
Given an execution $e$, let $\mathit{minUnmatched}(e,p)$ be the set of unmatched send actions in $e$ which are minimal in the causal relation of $tr(e)$ among unmatched send actions with destination $p$, i.e., $\mathit{minUnmatched}(e,p)$ is the set of unmatched send actions $\send{i}{p',p,v}$ in $e$ such that for every other unmatched send action $\send{j}{p'',p,v}$ in $e$ we have that $\send{i}{p'',p,v}\not\leadsto_{tr(e)}\send{i}{p',p,v}$.

\begin{theorem}
A $k$-synchronizable system $\mathcal{S}$ reaches an unspecified reception configuration under the asynchronous semantics iff there exists some $p\in \<Pids>$ such that the $k$-synchronous semantics of $\mathcal{S}$ admits an execution $e$ to a $(p,V)$-receiving state and 
\begin{align*}
\{v: \exists \send{i}{p',p,v}\in \mathit{minUnmatched}(e,p)\}\setminus V \neq \emptyset.
\end{align*}
\end{theorem}
\begin{proof}
For the only-if direction, let $e$ be an execution in $\asynchExec{\mathcal{S}}$ to an unspecified reception configuration $(\vec{l},\vec{b})$. Then, there exists $p\in\<Pids>$ such that $\vec{l}_p$ is $V$-receiving, for some $V\in\<Vals>$, and $v_p\not\in V$, where $v_p$ is the head of $\vec{b}_p$ (the first message to be dequeued). Therefore, $e$ contains an unmatched send action $\send{i}{p',p,v_p}$ which is also the first among unmatched send actions with destination $p$ (otherwise, $v_p$ would not be the first message in the buffer of $p$). Therefore, $\send{i}{p',p,v_p}\in \mathit{minUnmatched}(e,p)$. By $k$-synchronizability, there exists a permutation $e'$ of $e$ that belongs to $\synchExec{\mathcal{S}}{k}$. By Lemma~\ref{lem:undist}, $e'$ is an execution to a configuration $(\vec{l},B)$, for some $B$, which is $(p,V)$-receiving. Since $e$ and $e'$ have the same trace, we get that $\send{i}{p',p,v_p}\in \mathit{minUnmatched}(e',p)$ as well, which finishes the proof of this direction.

For the if direction, assume that the $k$-synchronous semantics of $\mathcal{S}$ admits an execution $e$ as above. Let $\send{i}{p',p,v}$ be an action in $\mathit{minUnmatched}(e,p)$ such that $v\not\in V$. Because $\send{i}{p',p,v}$ is minimal among unmatched send actions with destination $p$, there exists a permutation $e'$ of $e$ where it is the first unmatched send with destination $p$. As a direct consequence of the definitions, we get that $e$ is an execution to an unspecified reception configuration.
\end{proof}

\section{Proofs of Section~\ref{sec:characterizations}}\label{asec:characterizations}

\begin{theorem}\label{lem:cg_k}
A trace $t$ satisfying causal delivery is $k$-synchronous if{f} every cycle in its conflict-graph is good and of size at most $k$.
\end{theorem}
\begin{proof}
For the only if direction, $t$ is the trace of an execution $e\in \synchExec{\mathcal{S}}{k}$ for some system $\mathcal{S}$. The execution $e$ is obtained through a sequence of \textsc{$k$-exchange} transitions. The set of actions of every node $v$ of $CG_t$ appear in the label of a single such transition. Moreover, for every cycle in $CG_t$, the actions corresponding to the nodes in this cycle occur in the label of a single \textsc{$k$-exchange} transition. Therefore, every cycle in $CG_t$ is good and of size at most $k$.

For the if direction, we first show that any strongly-connected component $C$ of $CG_t$ is $k$-synchronous. Since all the cycles in $CG_t$ are of size at most $k$, we get that $C$ contains at most $k$ nodes. The case of strongly-connected components formed of a single node $v$ is trivial. The actions corresponding to $v$ are either a matching pair of send/receive actions, which can be generated through a \textsc{$1$-exchange} transition, or an unmatched send, which can also be generated through a \textsc{$1$-exchange} transition. Now, let $C$ be formed of a set of nodes 
$v_1$,$\ldots$, $v_n$, for some $2\leq n\leq k$ such that $s_i\in act(v_i)$, for all $1\leq i\leq n$. 
W.l.o.g., assume that the indexing of the nodes in $C$ is consistent with the edges labeled by SS (note that there is no cycle formed only of edges labeled by SS), i.e., for every $1\leq i_1<i_2\leq n$, $C$ doesn't contain an edge labeled by SS from $i_2$ to $i_1$, and for every $1\leq i_1<i_2<i_3\leq n$, if $\<proc>(s_{i_1})=\<proc>(s_{i_3})$, then $\<proc>(s_{i_1})=\<proc>(s_{i_2})$. Let $i_1$,$\ldots$,$i_m$ be the maximal subsequence of $1$,$\ldots$,$n$ such that $r_{i_j}\in act(v_i)$ for every $1\leq j\leq m$. 
We have that $C$ is the trace of the execution $e=s_1\cdot\ldots\cdot s_n\cdot r_{i_1}\cdot\ldots\cdot r_{i_m}$. The fact that all sends can be executed before the receives is a consequence of the fact that $C$ doesn't contain edges labeled by $RS$. Then, the order between receives is consistent with the one between sends because $C$ satisfies causal delivery. By definition, $e$ is the label of an \textsc{$n$-exchange} transition, and therefore, $C$ is $k$-synchronous. 

To complete the proof we proceed by induction on the number of strongly connected components of $CG_t$. The base case is trivial. 
 For the induction step, assume that the claim holds for every trace whose conflict-graph has at most $n$ strongly connected components, and 
 let $t$ be a trace with $n+1$ strongly connected components.  
 Let $C$ be a strongly connected component of $t$ such that
$C$ has no outgoing edges towards another strongly connected component of $t$.
%
 By the definition of the conflict-graph, $t$ is the trace of an execution $e=e'\cdot e''$ where $e''$ contains all the actions corresponding to the nodes of $C$. 
 We have shown above that $e''$ is $k$-synchronous, and by the induction hypothesis, $e'$ is also $k$-synchronous. Therefore, $e$ is $k$-synchronous.
\end{proof}

\section{Proofs of Section~\ref{sec:verif}}\label{asec:verif}

\subsection{Borderline Synchronizability Violations}

\begin{lemma}
Let $e$ be a borderline violation to $k$-synchronizability of $\mathcal{S}$. Then, $e = e'\cdot r$ for some $e'\in (S_{id}\cup R_{id})^*$ and $r\in R_{id}$.
\end{lemma}
\begin{proof}
Assume by contradiction that $e=e'\cdot s$ for some $e'\in (S_{id}\cup R_{id})^*$ and $s\in S_{id}$. By definition, $CG_{tr(e)}$ contains no outgoing
edge from the node representing $s$, which implies that any cycle of $CG_{tr(e)}$ is already contained in $CG_{tr(e')}$. This is a contradiction to 
the fact that $e$ is borderline.
\end{proof}

\begin{lemma}\label{lem:critical}
Let $e = e'\cdot r$, for some $e'\in (S_{id}\cup R_{id})^*$ and $r\in R_{id}$, be a borderline violation to $k$-synchronizability of $\mathcal{S}$. 
Then, the node $v$ of $CG_{tr(e)}$ representing $r$ (and the corresponding send) is a critical node of every cycle of 
$CG_{tr(e)}$ which is bad or of size bigger than $k$. 
\end{lemma}
\begin{proof}
Let $v_0,v_1,\ldots,v_n,v_0$ be a cycle of $CG_{tr(e)}$ which is bad or of size bigger than $k$. We first show that $v=v_i$ for some $0\leq i\leq n$. 
Assume by contradiction that this is not the case. Then, the execution $e'$ is already a violation to $k$-synchronizability which violates the assumption that $e$ is borderline.

For the following, w.l.o.g., we assume that $v=v_0$. Since $r$ is the last action of $e$, the only outgoing edge of $v$ is an edge labeled by $SX$ with $X\in \{S,R\}$. Therefore $(v,v_1)$ is an $SX$ edge. 
Assuming by contradiction that the edge $(v_n,v)$ is labeled by $YS$ for some $Y\in \{S,R\}$ implies that $e'$ is already a $k$-synchronizability violation, which contradicts the hypothesis.
\end{proof}

\subsection{Simulating Borderline Violations on the Synchronous Semantics}

We define a system $\mathcal{S'}$ which ``delays'' the reception of exactly one message, which is chosen nondeterministically,
by redirecting it to an additional process $\pi$ which relays it to the original destination at a later time. We show that
the synchronous semantics of $\mathcal{S'}$ ``simulates'' a permutation of every borderline violation of 
$\mathcal{S}$. 

Formally, given $\mathcal{S}=((\<Lsts>_p,\delta_p,l_p^0)\mid p\in\<Pids>)$, we define $\mathcal{S'}=((\<Lsts>_p,\delta'_p,l_p^0)|p\in\<Pids>\cup\{\pi\})$ where
\begin{itemize}
	\item every send of a process $p$ can be non-deterministically redirected to the process $\pi$ (the message payload contains the destination process), i.e., 
	\begin{align*}
	&\delta'_p(l,\senda{p,\pi,(q,v)}) = \delta'_p(l,\senda{p,q,v})\mbox{, and} \\ 
	&\delta'_p(l,a)=\delta_p(l,a)\mbox{ for all $p\in\<Pids>$, $l\in \<Lsts>_p$, and $a\not\in \{\senda{p,\pi,v}| p\in\<Pids>, v\in\<Vals>\}$}
	\end{align*}
	\item the process $\pi$ stores the received message in its state and relays it later, i.e., $\<Lsts>_\pi=\{l_\pi^0,l_f\}\cup\{(q,v)\mid q\in\<Pids>, v\in\<Vals>\}$, and
	for all $q\in\<Pids>$ and $v\in\<Vals>$, 
	\begin{align*}
	&\delta'_p(l_\pi^0,\reca{\pi,(q,v)}) = \{(q,v)\} \mbox{ and }
	\delta'_p((q,v),\senda{\pi,q,v})=l_f
	\end{align*}	
\end{itemize}

\begin{lemma}
Let $e=e_1\cdot \send{i}{p,q,v}\cdot e_2\cdot \rec{i}{q,v}$ be a borderline violation to $k$-synchronizability of $\mathcal{S}$. Then, $\synchExec{\mathcal{S'}}{k}$ contains an execution $e'$ of the form: 
\begin{align*}
e'=e_1'\cdot \send{i}{p,\pi,(q,v)}\cdot \rec{i}{\pi,(q,v)}\cdot e_2'\cdot \send{j}{\pi,q,v}\cdot \rec{j}{q,v}
\end{align*}
such that $e_1'\cdot \send{i}{p,q,v} \cdot e_2'$ is a permutation of $e_1\cdot \send{i}{p,q,v}\cdot e_2$.
\end{lemma}
\begin{proof}
A direct consequence of the definition of $\mathcal{S'}$ is that $e'\in\asynchExec{\mathcal{S'}}$. We show that the trace of $e'$ is $k$-synchronous. The conflict graph of $tr(e')$ can be obtained from the one of $tr(e)$ as follows:
\begin{itemize}
	\item the node $v$ representing the pair of actions $\{\send{i}{p,q,v},\rec{i}{q,v}\}$ is replaced by two nodes $v'$ and $v''$ representing $\{\send{i}{p,\pi,(q,v)}, \rec{i}{\pi,(q,v)}\}$ and $\{\send{j}{\pi,q,v}, \rec{j}{q,v}\}$, respectively,
	\item for every $SX$ edge from $v$ to a node $w$ in $CG_{tr(e)}$, where $X\in\{S,R\}$, there exists an $SX$ edge from $v'$ to $w$ in $CG_{tr(e')}$,
	\item $v'$ is connected to $v''$ by an $RS$ edge,
	\item there is no outgoing edge from $v''$.
\end{itemize}
Since all the cycles of $CG_{tr(e)}$ that are bad or exceed the size $k$ pass trough $v$, we get that $CG_{tr(e')}$ contains no such cycle.
Therefore, $tr(e')$ is $k$-synchronous. 

This implies that $\synchExec{\mathcal{S'}}{k}$ contains a permutation of $e_1\cdot \send{i}{p,\pi,(q,v)}\cdot \rec{i}{\pi,(q,v)}\cdot e_2\cdot \send{j}{\pi,q,v}\cdot \rec{j}{q,v}$. Since there is no outgoing edge from $v''$, there exists such a permutation that ends in $\send{j}{\pi,q,v}\cdot \rec{j}{q,v}$ which concludes the proof.
\end{proof}

\subsection{Excluding Executions Violating Causal Delivery}\label{asec:causal-delivery}

\begin{figure}[t]
\begin{center}
\centering
\begin{lstlisting}
function cone: $2^{\<Pids>}$
function receiver: $\<Pids>\cup \{\bot\}$

rule $s_1\cdot\ldots\cdot s_n\cdot r_1\cdot\ldots\cdot r_m$:
  if ( $\exists i.\ \<proc>(s_i)=\pi$ )
    pass
  if ( $\exists i.\ s_i=\send{\_}{p,\pi,(q,v)}$ )
    cone := $\{p\}$
    receiver := $q$
  forall j with $1\leq j\leq n$
    if ( $p\in \mbox{cone} \land \exists k.\ s_j\match r_k\land (\exists i.\ \<dest>(s_i) = \pi\implies (\<proc>(s_i)=\<proc>(s_j)\land i < j ))$ )
      cone := $\mbox{cone} \cup \{\<dest>(s_j)\}$
      assert $\<dest>(s_j) \neq \mbox{receiver}$
\end{lstlisting}
\end{center}
%
%
\caption{The monitor $\mathcal{M}_{\mathit{causal}}$. Initially, {\tt receiver} is $\bot$, and ${\tt cone}=\emptyset$.}
\label{fig:mon_causal}
\end{figure}

The monitor $\mathcal{M}_{\mathit{causal}}$ is essentially a labeled transition system whose transitions are labeled by sequences of actions in $S_{id}^*\cdot R_{id}^*$ corresponding to $k$-exchange transitions of the synchronous semantics. For readability, we define it as an abstract state machine in Figure~\ref{fig:mon_causal}. $\mathcal{M}_{\mathit{causal}}$ tracks the set of processes ${\tt cone}$ who perform a send that is causally after the send to $\pi$, and checks that the message they sent is not received by the same process to whom $\pi$ must send a message. It performs this check before $\pi$ sends a message and therefore, any failure will correspond to an execution which is not possible in $\mathcal{S}$ (violating causal delivery). 

The set of executions of the $k$-synchronous semantics of $\mathcal{S'}$ for which $\mathcal{M}_{\mathit{causal}}$ doesn't go to an error state (the assertion in $\mathcal{M}_{\mathit{causal}}$ passes at every step) is denoted by $\mathcal{S}_k' \paral\mathcal{M}_{\mathit{causal}}$. The following result shows that every such execution is correct with respect to the asynchronous semantics of $\mathcal{S}$.

\begin{lemma}
For every execution $e\in \mathcal{S}_k' \paral\mathcal{M}_{\mathit{causal}}$, we have that $\sigma(e)\in \asynchExec{\mathcal{S}}$.
\end{lemma}

The reverse of the lemma above is also true, modulo permutations.

\begin{lemma}
For every borderline violation $e\in \asynchExec{\mathcal{S}}$ to $k$-synchronizability, there exists an execution $e'\in \mathcal{S}_k' \paral\mathcal{M}_{\mathit{causal}}$, such that $\sigma(e')$ is a permutation of $e$.
\end{lemma}

\subsection{Detecting Synchronizability Violations}\label{asec:violations}

\begin{figure}
\begin{minipage}{6.5cm}
\begin{lstlisting}
function conflict: $\<Pids>\cup\{\bot\}$
function lastIsRec: $\mathbb{B}$
function sawRS: $\mathbb{B}$
function count: $\mathbb{N}$

rule $\send{i}{p,\pi,(q,v)}\cdot \rec{i}{\pi,(q,v)}$:
  conflict := $p$
  count := k

//for every $i$, $\<dest>(s_i)\neq \pi$ and $\<proc>(s_i)\neq \pi$
rule $s_1\cdot\ldots\cdot s_n\cdot r_1\cdot\ldots\cdot r_m$:
  for i = $1$ to $n$
    if ( * $\land$ $\exists j.\ s_i \match r_j \land \mbox{conflict} \in \{\<proc>(s_i),\<dest>(s_i)\} $ )
      if ( * )
        conflict := $\<proc>(s_i)$
        if ($\mbox{lastIsRec}$)
          sawRS = true
        lastIsRec := false
      else 
        conflict := $\<dest>(s_i)$
        lastIsRec := true
      count--
    if ( * $\land$ $\<proc>(s_i)=\mbox{conflict}\land \forall j.\ \neg s_i\match r_j$ )
      count--
      lastIsRec := false

rule $\send{i}{\pi,q,v}\cdot \rec{i}{q,v}$:
  assert $\mbox{conflict} =q \implies (\mbox{count} > 0 \land \neg \mbox{sawRS})$
\end{lstlisting}
\end{minipage}
\hspace{5mm}
\begin{minipage}{5.5cm}
\includegraphics[width=5.5cm]{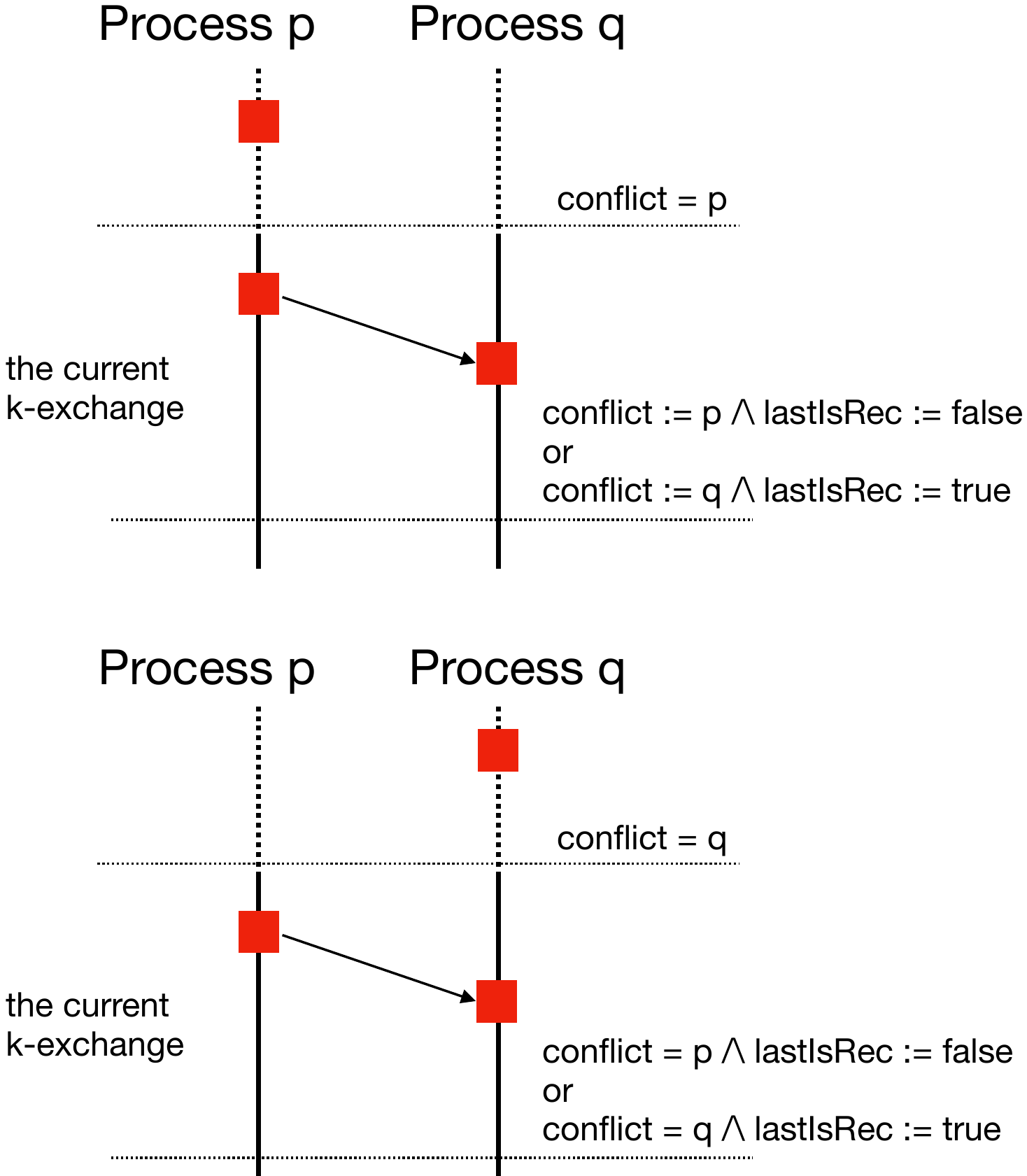}
\end{minipage}
\caption{The monitor $\mathcal{M}_{\mathit{viol}}(k)$. We use $\mathbb{B}$ to denote the set of Boolean values and $\mathbb{N}$ the set of natural numbers. Initially, {\tt conflict} and {\tt receiver} are $\bot$, while {\tt lastIsRec} and {\tt sawRS} are {\tt false}.}
\label{fig:mon_viol}
\end{figure}

Figure~\ref{fig:mon_viol} lists the definition of $\mathcal{M}_{\mathit{viol}}(k)$ as an abstract state machine. By the construction of $\mathcal{S'}$, we assume w.l.o.g., that both the send to $\pi$ and the send from $\pi$ are executed in isolation as an instance of $1$-exchange.
When observing the send to $\pi$, the monitor updates the variable ${\tt conflict}$, which in general stores the process executing the last action in the cycle, to $p$. Also, a variable ${\tt count}$, which becomes $0$ when the cycle has strictly more than $k$ nodes, is initialized to $k$. 

Then, for every $k$-exchange transition in the execution, $\mathcal{M}_{\mathit{viol}}(k)$ non-deterministically picks pairs of matching send/receive or unmatched sends to continue the conflict-graph path, knowing that the last node represents an action of the process stored in ${\tt conflict}$. The rules for choosing pairs of matching send/receive to advance the conflict-graph path are pictured on the right of Figure~\ref{fig:mon_viol} (advancing the conflict-graph path with an unmatched send doesn't modify the value of ${\tt conflict}$, it just decrements the value of ${\tt count}$). In principle, there are two cases depending on whether the last node in the path conflicts with the send or the receive of the considered pair. One of the two process involved in this pair of send/receive equals the current value of ${\tt conflict}$. Therefore, ${\tt conflict}$ can either remain unchanged or change to the value of the other process. The variable {\tt lastIsRec} records whether the current conflict-graph path ends in a conflict due to a receive action. If it is the case, and the next conflict is between this receive and a send, then {\tt sawRS} is set to {\tt true} to record the fact that the path contains an $RS$ labeled edge (leading to a potential bad cycle).

When $\pi$ sends its message to $q$, the monitor checks whether the conflict-graph path it discovered ends in a node representing an action of $q$. If this is the case, this path together with the node representing the delayed send forms a cycle. Then, if ${\tt sawRS}$ is ${\tt true}$, then the cycle is bad and if ${\tt count}$ reached the value $0$, then the cycle contains more than $k$ nodes. In both cases, the current execution is a violation to $k$-synchronizability.

\section{Proofs of Section~\ref{sec:decidability}}\label{asec:decidability}

\begin{theorem}\label{th:dec1}
For a finite-state system $\mathcal{S}$, the reachability problem under the $k$-synchronous semantics is decidable and PSPACE-complete.
\end{theorem}
\begin{proof}
A consequence of the fact that the product emptiness problem (checking if the product of a set of finite state automata has an empty language) is PSPACE-complete~\cite{DBLP:conf/focs/Kozen77}. The evolution of the $B$ component of a synchronous configuration and the set of messages sent during a $k$-exchange transition can be modeled using an additional labeled transition system that is composed with the processes in the system. 
\end{proof}

\begin{theorem}
The problem of checking $k$-synchronizability of a finite-state system $\mathcal{S}$ is decidable and PSPACE-complete.
\end{theorem}
\begin{proof}
Theorem~\ref{th:main-verif} and Theorem~\ref{th:dec1} imply that the problem is in PSPACE. Moreover, PSPACE-hardness follows from the fact that the product emptiness problem can be reduced to checking $1$-synchronizability. Given a set of finite state automata $A_1$, $\ldots$, $A_n$, we define a message passing system $\mathcal{S}$ containing one process $p_i$ for each automaton $A_i$, which ``simulates'' the product. Essentially, $p_1$ is obtained from $A_1$ by rewriting every transition label $a$ to $\senda{p_1,p_2,a_1}\cdot \reca{p_1,a_n}$, the process $p_i$ with $1<i<n$ is obtained from $A_i$ by rewriting every transition label $a$ to $\reca{p_i,a_{i-1}}\cdot \senda{p_i,p_{i+1},a_i}$, and $p_n$ is obtained from $A_n$ by rewriting every transition label $a$ to $\reca{p_n,a_{n-1}}\cdot \senda{p_n,p_1,a_n}$. This rewriting ensures that every transition of the product of $A_1\times\ldots\times A_n$ is simulated precisely by a sequence of sends/receives:
\begin{align*}
\senda{p_1,p_2,a_1}\cdot \reca{p_2,a_1}\cdot \senda{p_2,p_3,a_2}\cdot\ldots \cdot \reca{p_n,a_{n-1}} \senda{p_n,p_1,a_n}\cdot \reca{p_1,a_n}
\end{align*}
Note that every execution admitted by this system is $1$-synchronous. Augmenting this system with new states and transitions to ensure that it produces a violation of $1$-synchronizability exactly when each process $p_i$ is in a final state of $A_i$, leads to a system which is \emph{not} $1$-synchronizable iff the product $A_1\times\ldots\times A_n$ has a non-empty language. Therefore, the product emptiness problem is polynomial-time reducible to checking $1$-synchronizability.
\end{proof}

\begin{theorem}
For a flow-bounded system $\mathcal{S}$, the problem of checking if there exists some $k$ such that $\mathcal{S}$ is $k$-synchronizable, is decidable.
\end{theorem}
\begin{proof}
First, assume that there exists an execution $e$ of $\mathcal{S}$ such that the corresponding conflict graph contains a bad cycle. Then, $\mathcal{S}$ is not $k$-synchronizable for $k = |e|$ (where $|e|$ denotes the length of $e$), and finding this $k$ through a procedure that checks $k$-synchronizability for increasing values of $k$ is clearly possible.

Now, assume that $\mathcal{S}$ admits no such execution. Then, every execution $e$ of $\mathcal{S}$ can be permuted to a $k$-synchronous execution $e'$, for some $k$ (the lack of conflict-graph cycles with an $RS$ labeled edge implies that the execution can be permuted to a sequence of $k$-exchange transition labels). Let $K$ be a constant such that every process in $\mathcal{S}$ is $K$-receive bounded and $K$-send bounded (this constant exists because $\mathcal{S}$ is flow-bounded). We get that the number of consecutive receives in $e'$ is bounded by $K\times |\<Pids>|$, and the number of consecutive sends by $K\times |\<Pids>|$. Otherwise, there would exist a process that performs more than $K$ consecutive receives or more than $K$ consecutive sends before a receive, which contradicts the definition of $K$. Therefore, $e'$ can be executed by a sequence of $K\times |\<Pids>|$-exchange transitions.
\end{proof}

\end{document}